\documentclass[journal]{IEEEtran}
\usepackage{algorithm,algorithmic}
\usepackage{latexsym}
\usepackage{cite}
\usepackage{CJK}
\DeclareMathAlphabet{\mathpzc}{OT1}{pzc}{m}{it}
\usepackage{bbm}
\usepackage{dsfont}
\usepackage{yfonts}
\usepackage{amsmath,amsthm}
\usepackage{amssymb}
\usepackage{amsfonts}
\usepackage{mathrsfs}
\usepackage{graphicx}
\usepackage{subcaption} 
\usepackage{tabularx}
\usepackage{epstopdf}
\usepackage{color}
\newtheorem{theorem}{Theorem}

\newtheorem{lemma}{Lemma}
\theoremstyle{remark}

\graphicspath{ {Images/} } 
\IEEEoverridecommandlockouts

\begin{document}

\title{Coverage Probability Analysis Under Clustered Ambient Backscatter Nodes}
\author{ Dong~Han  and Hlaing~Minn
\thanks{ The authors are with the Department of Electrical and Computer Engineering, University of Texas at Dallas, Emails:  \{dong.han6, hlaing.minn\}@utdallas.edu. } 
}

\maketitle

\begin{abstract}
	In this paper, we consider a new large-scale communication scheme where randomly distributed AmBC nodes are involved as secondary users to primary transmitter (PT) and primary receiver (PR) pairs. The secondary communication between a backscatter transmitter (BT) and a backscatter receiver (BR) is conducted by the BT's reflecting its corresponding PT's signal with different antenna impedances, which introduces additional double fading channels and potential inter-symbol-interference to the primary communications. Thus, at a typical PR, the backscatter signals are regarded as either decodable signals or interference. Assuming the locations of PTs form a Poisson point process and the locations of BTs form a Poisson cluster process, we derive the SINR and SIR based coverage probabilities for two network configuration scenarios. Numerical results on the coverage probabilities indicate the possibility to involve a large amount of AmBC nodes in existing wireless networks.
\end{abstract}	


\section{Introduction} 
\subsection{Background and Motivation}
The backscatter mechanism enables backscatter transmitters (BTs) to have a simple structure consisting of no active radio frequency (RF) component, which is strongly favored by the Internet-of-things (IoT) application scenarios where many power-limited devices need to be connected.
There are three configurations of backscatter communication systems, namely, monostatic backscatter (where the carrier emitter and receiver are co-located), bistatic backscatter (where the carrier emitter and receiver are geographically separated) and ambient backscatter (where ambient RF signals are used as carriers)\cite{van2018ambient}.
For conventional monostatic backscatter techniques such as radio frequency identification (RFID), the transmitter can only passively transmit to the reader when being inquired\cite{boyer2014invited}, which limits the application scenarios of this technology. 
Recently, the developments in wireless power transfer (WPT)\cite{zeng2017communications} and signal detection techniques, such as successive interference cancellation (SIC)\cite{weber2007transmission,ding2014performance,wildemeersch2014successive}, reignite the backscatter communication. Consequently, two research thrusts of backscatter communications are beginning to be eagerly investigated, i.e., the Wireless Powered Backscatter Communication (WPBC)\cite{clerckx2017wirelessly,han2017wirelessly} and the Ambient Backscatter Communication (AmBC)\cite{liu2013ambient,yang2018cooperative}.

The WPBC system has a bistatic configuration, where one or more carrier emitters transmit sinusoidal continuous waves (CWs) for the BT to 1) operate with the harvested energy from a portion of the CWs and 2) modulate its information bits with the other portion of the CWs by backscattering the CWs with different antenna impedances. Then, the backscatter receiver (BR) receives the backscattered (modulated) signals from the BT and detect the `0' or `1' information based on the average symbol energy level. 

\begin{figure}[!tbp]
	\centering
	\includegraphics[width = 0.95\linewidth]{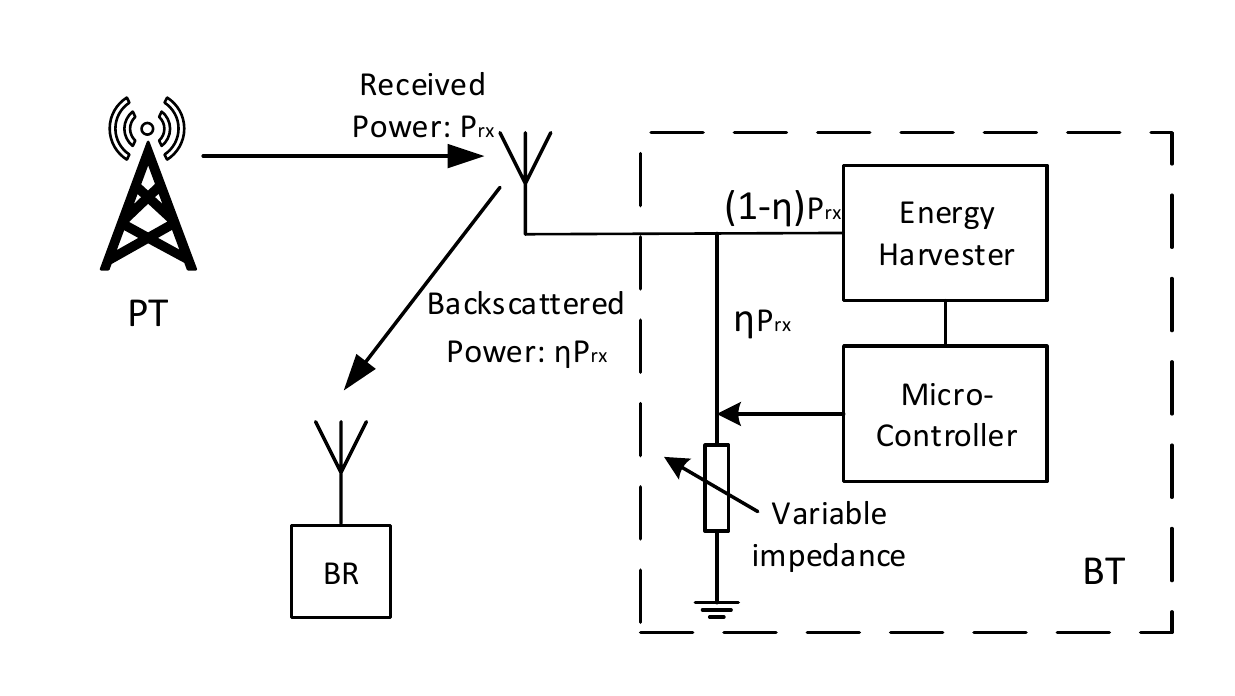} 
	\caption{Ambient backscatter scheme} 	
	\label{backscatter}
	\hfill
\end{figure}

AmBC was proposed to enable devices to communicate by backscattering ambient RF signals\cite{liu2013ambient}. As shown in Fig.~\ref{backscatter}, the BT harvests energy from the ambient signal and quickly transmits its own information bits to the corresponding BR by changing the impedance of its antenna in the presence of the ambient signal. For instance, the BT transmits '0' by setting a high antenna impedance and transmits '1' by adjusting to a low antenna impedance so that the BR can distinguish the different backscattered signal energy levels from the BT.

Recent studies about AmBC mainly focus on the signal detection perspective and most of them consider a single BT-BR pair and a PT (which emits the ambient RF signal) although multiple antennas are involved\cite{yang2018cooperative,wang2016ambient,qian2017noncoherent}. Concerning the rapid growth of IoT devices, we believe investigating the scalability of AmBC is also a crucial issue. Thus, this paper considers to include the AmBC nodes in conventional large-scale wireless communication networks where the locations of primary transmitters (PTs) such as users and devices form a Poisson point process (PPP) and the BTs cluster around each PT and backscatter the corresponding PT's signal for their own data transmission. 

\subsection{Related Works and Our Contributions}
Stochastic geometry based approaches have been realized to be efficient and tractable for analyzing complex heterogeneous networks (HetNet) \cite{elsawy2013stochastic}. With some assumptions to the distribution (such as PPP) of the node locations, the system performance of a HetNet can be expressed by quickly computable integrals with a small number of parameters \cite{andrews2011tractable}. Recent studies have found that simply using a PPP based geometric model is not rich enough to analyze the increasingly complex HetNet, yet the Poisson cluster process (PCP) based analysis is more capable \cite{saha20183gpp}. Some quantitative properties of PCP and PCP based device-to-device (D2D) network can be found in \cite{saha20183gpp, afshang2017nearest, afshang2017nearest2, joshi2018coverage}. Specifically, the coverage probabilities of several HetNet configurations based on the 3rd generation partnership project (3GPP) model were studied in \cite{saha20183gpp}. The nearest neighbor and contact distance distributions for Mat\'ern and Thomas cluster processes were investigated in \cite{afshang2017nearest} and \cite{afshang2017nearest2}, respectively. The authors in \cite{joshi2018coverage} derived the approximate coverage probability of a PCP-based D2D network with Nakagami-m fading channel.

Considering multiple randomly distributed carrier emitters (power beacons) and clustered passive devices, the authors in \cite{han2017wirelessly} proposed a large-scale WPBC network and derived the network coverage probability and capacity. The achievable rate region for a single-tag backscatter multiplicative multiple-access channel, where the receiver detects both the transmitter and the tag's signals, was derived in \cite{liu2018backscatter}. A hybrid transmission scheme that integrates AmBC and wireless powered communications was proposed in \cite{lu2018ambient} to trade off the hardware implementation complexity and the data transmission performance.

However, to the best of our knowledge, none of the existing studies have considered to apply a PCP model to an AmBC system, of which the analysis is different from those proposed in the references. Specifically, in an existing large-scale wireless network, newly deployed AmBC nodes will change the effective channel response between a primary transmitter (PT) and a primary receiver (PR). In this scenario, the backscattered signals can be regarded as either decodable signals or interference at a typical PR, which will affect the coverage probability of a typical PT. Therefore, we derive an analysis of the signal-to-interference-plus-noise ratio (SINR) and signal-to-interference ratio (SIR) based coverage probability at a typical PR. The contributions of this paper are summarized below.
\begin{itemize}
	\item We construct a new large-scale communication scenario where randomly distributed ambient backscatter nodes are involved as secondary users to primary transmitter (PT) and primary receiver (PR) pairs. The backscatter signals are regarded as either decodable signals or interference. We denote $\beta$ as the fraction of decodable backscattered signal power and propose an approach to estimate $\beta$.
	\item Considering the double-fading effect, we derive the SINR and SIR based coverage probabilities for two network configuration scenarios (i.e., all PTs are surrounded by active BTs in scenario-1 and only a typical PT is surrounded by active BTs in scenario-2) with two ranges of $\beta$ values. 
	\item Extensive numerical results are provided to evaluate the coverage probabilities in different scenarios with various key system parameters. Comparing with the benchmark scenario where no BT exists, the numerical results indicate the possibility and advantages to involve a large amount of AmBC nodes in existing wireless networks. 
\end{itemize}

This paper is organized as follows. The system model and signal representations are described in Section II. Then, we derive the coverage probabilities in Section III. Numerical simulation results and conclusions are provided in Section IV and Section V, respectively. 

{\color{black}\textit{Notations:} $B(\boldsymbol{c},r)$ represents a disk centered at $\boldsymbol{c}$ with radius $r$ and bold number $\boldsymbol{0}$ refers to the origin. $\mathbb{E}_\mathrm{X}(\cdot)$ denotes the expectation operator over $X$. 
}

\section{System Model} 	
\subsection{Spatial Distribution Models} 
\begin{figure}[!tbp]
	\centering
	\includegraphics[width = 1\linewidth]{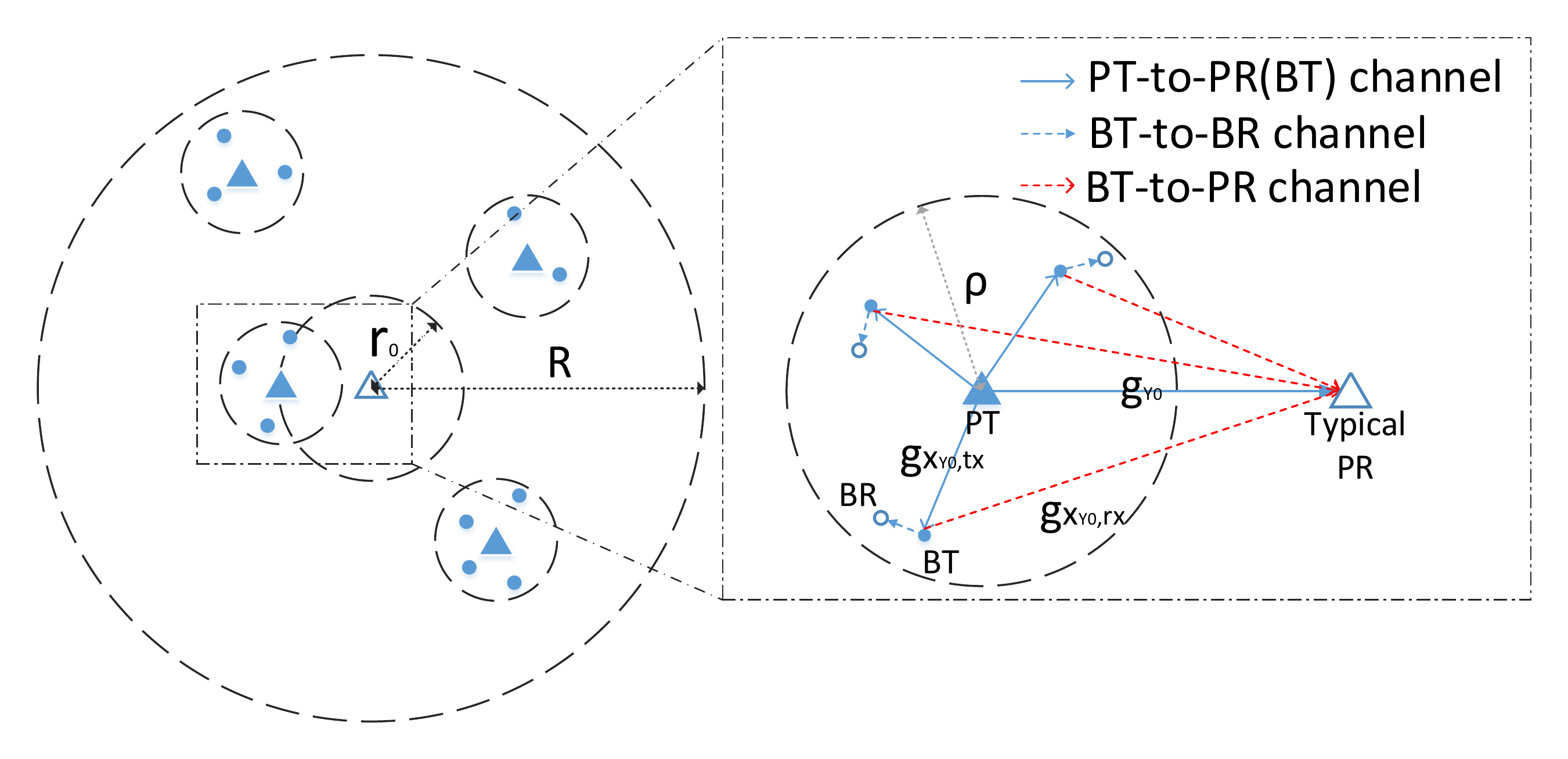} 
	\caption{Topology of the considered system (scenario-1)} 
	\label{system}
	\hfill
\end{figure}

{\color{black}The system we consider consists of two tiers (layers), where the first tier includes all the PTs and PRs, and the second tier includes all the BTs and BRs. Under such system, we study two BT deployment scenarios. For scenario-1, each PT is surrounded by a cluster of BTs, shown in Fig. \ref{system}. For scenario-2, only the typical PT is surrounded by a cluster of BTs.}
Since the SINR and SIR based coverage probability derivations of both scenarios are similar, we mainly focus on the SINR-based derivation of scenario-1 in the following.
Without loss of generality, we set a typical PR at the origin and its corresponding PT {\color{black}(i.e., the closest PT to the typical PR)} at coordinate $Y_0 = (r_0,0)$\footnote{For a homogeneous PPP with density $\lambda$, the distance between an arbitrary (typical) point and its closest point is Rayleigh distributed with the scale parameter ${1}/{\sqrt{2 \pi \lambda}}$ \cite{andrews2011tractable}. Thus, coverage probabilities based on a fixed distance $r_0$ can be extended to a general coverage probability by de-conditioning them with the distribution of $r_0$.}. The locations of other PTs which cause interference to the typical PR form a PPP $\Phi_\mathrm{P} = \{Y_j\}$, where $Y_j \in \mathbb{R}^2,\ j = 1,2,\ldots, M$, represents the coordinate of the $j$th PT, with a constant density $\lambda_\mathrm{P}$ in the ring-shape region $B(\boldsymbol{0},R)-B(\boldsymbol{0}, r_0)$. 
The locations of the BTs\footnote{The distribution of BRs is not considered since it does not affect the coverage probability.} form a Mat\'ern cluster process\cite{afshang2017nearest} represented by $\bigcap_{j=1}^{M} \Phi_\mathrm{B}(Y_j)$ in scenario-1 (where $\Phi_\mathrm{B}(Y_j) = \{X_\mathrm{Y_j}\}$, with $X_\mathrm{Y_j} \in \mathbb{R}^2$ representing the coordinate of a BT in the disk $B(Y_j,\rho)$), and a PPP represented by $\Phi_\mathrm{B}(Y_0)$ in scenario-2. For a compact expression, the distance between a PT at $Y$ and the typical PR at the origin is denoted by $r_\mathrm{Y}$. Similarly, we use $r_\mathrm{X_{Y},tx}$ ($r_\mathrm{X_{Y},rx}$) to represent the distance between the PT at $Y$ (the typical PR at the origin) and the offspring BT at $X_\mathrm{Y}$.

In addition, we assume that the density of PTs is much smaller than the BTs' density, i.e., $\lambda_P \ll \lambda_B$, such that the distances between BTs and their non-parent PTs are relatively large in average sense, resulting in much severer path losses than those between BTs and their parent PTs. Thus, we further assume that the information signals sent from the parent PT is the only {\color{black}RF} power source of its offspring BTs and leave the analysis of multiple power sources for future study.

\subsection{Signal Communication Model} 
We denote the transmit symbol of the PT located at $Y$ at time $t$ by $\sqrt{P_\mathrm{tx}} s_\mathrm{Y}(t)$, {\color{black}where $P_\mathrm{tx}$ is the constant transmit power} and $s_Y(t) \sim \mathcal{CN}(0,1)$ is the normalized complex Gaussian distributed symbol. The backscatter symbol of the BT located at $X_\mathrm{Y}$ at time $t$ is denoted by $b_\mathrm{X_Y}(t)$. Since a BT only reflects the ambient signal using two impedance levels, we assume that the BTs' symbols are independent and Bernoulli-distributed with equal probability, i.e., $b_\mathrm{X_Y}(t) \sim \mathrm{Bernoulli} \left(\frac{1}{2}\right),\ \forall X_\mathrm{Y}$. 
As shown in Fig.~\ref{backscatter}, 
the reflection coefficient is $\eta \in [0,1]$, which means $\eta P_\mathrm{rx}$ of the received power $P_\mathrm{rx}$ is backscattered by the BT and $(1-\eta)P_\mathrm{rx}$ of the power is harvested by the BT for modulation and control purpose. Besides, we simply assume backlogged transmissions for both tiers so that the BTs and BRs can always be active based on the harvested energy. The AmBC throughput maximization problem regarding to mode switching policy has been investigated in \cite{wen2019throughput}. 

Furthermore, we assume independent and identically distributed (i.i.d.) Rayleigh fading with average power gain of $1/\mu$ and path loss with exponent $\alpha$ over all channels. 
In particular, suppose the Rayleigh fading channel power gains between the PT at $Y$ and the typical PR at the origin, the PT at $Y$ and the BT at $X_\mathrm{Y}$, and the BT at $X_\mathrm{Y}$ and the typical PR are $g_\mathrm{Y}$, $g_\mathrm{X_Y,tx}$ and $g_\mathrm{X_Y,rx}$, respectively, each with exponential distribution with parameter $\mu$, denoted by $\exp(\mu)$. Then, the channel response can be written as ${h}_\mathrm{Y} = \sqrt{ {g}_\mathrm{Y} } e^{j\theta_\mathrm{Y}},\ {h}_\mathrm{X_Y,tx}  = \sqrt{{g}_\mathrm{X_Y,tx}} e^{j\theta_\mathrm{X_Y,tx}} \text{ and } {h}_\mathrm{X_Y,rx}  = \sqrt{{g}_\mathrm{X_Y,rx}} e^{j\theta_\mathrm{X_Y,rx}}$, respectively, each with zero-mean-complex-Gaussian distribution, where $\theta_\mathrm{Y}$, $\theta_\mathrm{X_Y,tx}$ and $\theta_\mathrm{X_Y,rx}$ represent the zero-mean uniformly distributed channel phases. To avoid the singularity at the origin, the path loss is expressed as $L(r) = (1+r^\alpha)^{-1}$ for a distance $r$.
Thus, we denote the received signal at the typical PR as the summation of several signals:
\begin{equation}
\begin{split}
y(t) = {s}_\mathrm{PT}(t) + {s}_\mathrm{BT}(t) + {I}_\mathrm{PT}(tx) + {I}_\mathrm{BT}(t) + n(t)
\end{split}
\end{equation}
where $n(t) \sim \mathcal{CN}(0,\sigma^2)$ is the complex Gaussian noise at PR, ${s}_\mathrm{PT}(t),\ {s}_\mathrm{BT}(t),\ {I}_\mathrm{PT}(t)$, and ${I}_\mathrm{BT}(t)$ represent the signals from the typical PT at $Y_0$, from the offspring BTs of the typical PT, from the atypical PTs, and from the offspring BTs of the atypical PTs, respectively. Particularly, we have
\begin{align}
& {s}_\mathrm{PT}(t) = {h}_\mathrm{Y_0} \sqrt{ L(r_0) P_\mathrm{tx}} s_\mathrm{Y_0}(t-\tau_\mathrm{Y_0}),\\
& {s}_\mathrm{BT}(t) = \underset{X_\mathrm{Y_0} \in \Phi_\mathrm{B}(Y_0) }{\sum} z_\mathrm{X_{Y_0}} \sqrt{\eta P_\mathrm{tx}} s_\mathrm{Y_0}(t-\tau_\mathrm{X_{Y_0}}) b_\mathrm{X_{Y_0}}  
\end{align}
\begin{align}
& {I}_\mathrm{PT}(t) = \underset{Y \in \Phi_\mathrm{P}}{\sum}  {h}_\mathrm{Y} \sqrt{ L(r_\mathrm{Y}) P_\mathrm{tx}} s_Y(t-\tau_\mathrm{Y}), \\
\begin{split}
&{I}_\mathrm{BT}(t)  =  \underset{Y \in \Phi_\mathrm{P} }{\sum} \ \underset{X_\mathrm{Y} \in \Phi_\mathrm{B}(Y) }{\sum} z_\mathrm{X_Y} \sqrt{ \eta P_\mathrm{tx}  }  s_\mathrm{Y}(t-\tau_\mathrm{X_Y}) b_\mathrm{X_Y}
\end{split}
\end{align}
where $z_\mathrm{X_Y} \triangleq {h}_\mathrm{X_Y,tx} {h}_\mathrm{X_Y,rx} \sqrt{  L(r_\mathrm{X_Y,tx}) L(r_\mathrm{X_Y,rx})}$, $\tau_\mathrm{Y}$ and $\tau_\mathrm{X_Y}$ are the time delays of the PT to PR path (direct path) and the PT-BT-PR path (backscatter path), respectively, and the subscripts $Y$ and $X_\mathrm{Y}$ indicate the locations of the PT and BT. We note that since the backscatter symbol $b_\mathrm{X_Y}(t)$ has a much larger symbol duration than the primary symbol $s_\mathrm{Y}(t)$ (i.e., $b_\mathrm{X_Y}(t)$ is constant in many successive symbols of $s_\mathrm{Y}(t)$) \cite{liu2013ambient}, the time index of and the time delay encountered by $b_\mathrm{X_Y}(t)$ are neglected.

\section{Analysis of Coverage Probability} 

\subsection{Signal and Interference Power} 
For a compact expression, we use vectors $\boldsymbol{z}_\mathrm{Y} = [\ldots, z_\mathrm{X_Y}, \ldots]^\mathrm{T}$ and $\boldsymbol{b}_\mathrm{Y} =  [\ldots, b_\mathrm{X_Y}, \ldots]^\mathrm{T} $ to represent the $z_\mathrm{X_Y}$'s and the $b_\mathrm{X_Y}$'s in the cluster centered at $Y$, respectively. Next, we can write the power of  ${I}_\mathrm{BT}(t)$ conditioned on the interference channels, path losses and the backscattered symbols as 
\begin{equation}
\begin{split}
\tilde{\mathcal{I}}_\mathrm{BT} = {\mathbb{E}} _{s_Y} \left[ \left| {I}_\mathrm{BT}(t) \right|^2 \right] = \eta P_\mathrm{tx}  \sum_{Y \in \Phi_P } \left| \boldsymbol{z}_\mathrm{Y}^\mathrm{T} \boldsymbol{b}_\mathrm{Y} \right|^2.
\end{split}
\label{interference0}
\end{equation}
However, the cross multiplication terms in (\ref{interference0}) make it difficult for further analysis, motivating us to decondition $\tilde{\mathcal{I}}_\mathrm{BT}$ on the channel phases\footnote{assuming that the channel phases change faster than amplitudes.} and backscattered symbols as 
\begin{equation}
\begin{split}
& \resizebox{0.9\hsize}{!}{$ \mathcal{I}_\mathrm{BT}  = \underset{Y \in \Phi_\mathrm{P} }{\sum} \underset{ \boldsymbol{b}_\mathrm{Y}, \theta}{\mathbb{E}} \left[ \tilde{\mathcal{I}}_\mathrm{BT}  \right] = \eta P_\mathrm{tx} \underset{Y \in \Phi_\mathrm{P} }{\sum} \underset{\boldsymbol{b}_\mathrm{Y} }{ \mathbb{E} } \left[  \boldsymbol{b}_\mathrm{Y}^\mathrm{T} \underset{\theta}{\mathbb{E}} \left[ \boldsymbol{z}_\mathrm{Y} \boldsymbol{z}_\mathrm{Y}^\mathrm{H} \right]\boldsymbol{b}_\mathrm{Y} \right] $} \\
& \quad \ \ = \frac{\eta P_\mathrm{tx} }{2}  \underset{Y \in \Phi_\mathrm{P} }{\sum} \underset{X_\mathrm{Y} \in \Phi_\mathrm{B} }{\sum}  g_\mathrm{X_Y,tx} g_\mathrm{X_Y,rx} L(r_\mathrm{X_Y,tx}) L(r_\mathrm{X_Y,rx}) 
\end{split} 
\label{I_BTPR}
\end{equation}
where the last equation follows from the facts that $\mathbb{E}_{\theta}\left[ \boldsymbol{z}_Y \boldsymbol{z}_Y^\mathrm{H} \right]$ is a diagonal matrix and $\mathbb{E}[b_{X_Y}^2(n)] = {1}/{2},\ \forall X_\mathrm{Y}$. Similarly, the powers of ${s}_\mathrm{PT}(t),\ {s}_\mathrm{BT}(t), \text{ and } {I}_\mathrm{PT}(t)$ can be represented as
\begin{align}
& \mathcal{S}_\mathrm{PT} = {g}_{Y_0} { L(r_0) P_\mathrm{tx}}, \label{S_PTPR} \\
& \mathcal{S}_\mathrm{BT} = \frac{\eta P_\mathrm{tx} }{2} \underset{X_{Y_0} \in \Phi_B(Y_0) }{\sum} {g}_\mathrm{X_{Y_0},tx} {g}_\mathrm{X_{Y_0},rx}  { L(r_\mathrm{X_{Y_0},tx}) L(r_\mathrm{X_{Y_0},rx})   }, \label{S_BTPR} \\
& \mathcal{I}_\mathrm{PT} = \underset{Y \in \Phi_P}{\sum} {g}_{Y} {L(r_Y) P_\mathrm{tx}}. \label{I_PTPR}
\end{align} 
We note that the mutual correlations among the signals ${s}_\mathrm{PT}(t),\ {s}_\mathrm{BT}(t),\ {I}_\mathrm{PT}(t), \text{ and } {I}_\mathrm{BT}(t)$ can be neglected if we decondition the correlations on the channel phases and use the fact that $s_Y(t)$'s are i.i.d. zero mean Gaussian. 

The typical PR aims to receive the typical PT's signal $s_{Y_0}(t)$, which is contained in ${s}_\mathrm{PT}(t)$ and ${s}_\mathrm{BT}(t)$. However, the time delays of the PT-BT-PR paths are larger than the time delay of the direct PT-PR path and are random due to the PCP formed by the BT locations, which may introduce different levels of inter symbol interference (ISI) at the typical PR. {\color{black}Therefore, ${s}_\mathrm{BT}(t)$ has a two-side effect, i.e., causing interference or enhancing the detection at the PR. Particularly, \cite{ruttik2018does} concludes that AmBC causes little interference to legacy systems in some deployment scenarios, and \cite{yang2018cooperative} indicates that detection performance can be enhanced with the cooperation of AmBC nodes. In this paper, we parameterize the two-side effect with $\beta \in [0,1]$, which denotes the fraction of the backscattered signal power that is not regarded as interference, and introduce an approach to estimate $\beta$ in Appendix A.}
Next, given $\beta \in [0,1]$, we can represent the SINR at the typical PR as
\begin{equation}
	\mathrm{SINR} = \frac{\mathcal{S}_\mathrm{PT} + \beta \mathcal{S}_\mathrm{BT} }{(1-\beta)\mathcal{S}_\mathrm{BT} + \mathcal{I}_\mathrm{PT} + \mathcal{I}_\mathrm{BT} + \sigma^2}
\end{equation}
and
\begin{equation}
\mathrm{SINR_u} = \frac{\mathcal{S}_\mathrm{PT} + \beta \mathcal{S}_\mathrm{BT} }{(1-\beta)\mathcal{S}_\mathrm{BT} + \mathcal{I}_\mathrm{PT} + \sigma^2 },
\end{equation}
for scenario-1 and scenario-2, respectively. 
Since $\mathrm{SINR} \le \mathrm{SINR_u}$, the coverage probability of scenario-2 (where BTs are only located around the typical PT) upper bounds the coverage probability of scenario-1 (where BTs are located around all PTs). We note that both scenarios can be interference-limited if the transmit signal-to-receive-noise ratio (TSRNR) ${P_\mathrm{tx}}/{\sigma^2}$ is large enough. Thus, we will also analyze the SIR based coverage probabilities for both scenarios and compare them with the SINR based results. 

\subsection{Coverage Probability Expression}

Denoting the SINR or SIR threshold at the typical PR as $\Gamma$, the coverage probability is defined as the probability that the SINR or SIR  is not less than the threshold:
\begin{equation}
\begin{split}
	\mathbb{P} \left(\mathrm{SINR} \ge \Gamma \right) = & \mathbb{P} ( \mathcal{S}_\mathrm{PT}  \ge  \left[ \Gamma(1-\beta) - \beta \right] \mathcal{S}_\mathrm{BT} \\
	& \qquad + \Gamma \mathcal{I}_\mathrm{PT} + \Gamma \mathcal{I}_\mathrm{BT} + \Gamma \sigma^2 ), \label{covprob_1_SINR} 
\end{split}
\end{equation} 

\begin{equation}
\begin{split}
	\mathbb{P} \left(\mathrm{SINR_u} \ge \Gamma \right) =  \mathbb{P} ( \mathcal{S}_\mathrm{PT} \ge &
	 \left[ \Gamma(1-\beta) - \beta \right] \mathcal{S}_\mathrm{BT} \\ &  + \Gamma \mathcal{I}_\mathrm{PT} + \Gamma \sigma^2), \label{covprob_1s_SINR} 
\end{split}
\end{equation}

\begin{equation}
\begin{split}
	 \mathbb{P} \left(\mathrm{SIR} \ge \Gamma \right) = \mathbb{P} ( \mathcal{S}_\mathrm{PT}  \ge  & \left[ \Gamma(1-\beta) - \beta \right] \mathcal{S}_\mathrm{BT} + \Gamma \mathcal{I}_\mathrm{PT} \\
	&  + \Gamma \mathcal{I}_\mathrm{BT} ), \label{covprob_1_SIR} 
\end{split}
\end{equation} 

\begin{equation}
\begin{split}
	\mathbb{P} \left(\mathrm{SIR_u} \ge \Gamma \right) = & \mathbb{P} ( \mathcal{S}_\mathrm{PT} \ge \left[ \Gamma(1-\beta) - \beta \right] \mathcal{S}_\mathrm{BT} + \Gamma \mathcal{I}_\mathrm{PT}  ). \label{covprob_1s_SIR} 
\end{split}
\end{equation}

To calculate the coverage probabilities, we will need the following lemmas:

\begin{lemma}
	To analyze the aggregated signal power at the typical PR, the clustered BTs can be approximately regarded as a virtual transmitter (VT) located at the center of the cluster. The VT's transmit power is the sum of backscattered signal powers $\tilde{P}_t$ of all BTs in the cluster, which is derived as
	\begin{equation}
	\begin{split}
		\tilde{P}_t & = \mathbb{E} \left[ \sum_{X_Y \in \Phi_B(Y) } g_\mathrm{X_Y, tx} L(r_\mathrm{X_Y,tx}) P_\mathrm{tx} \right] \\
		& \overset{\mathrm{(a)}}{=} \lambda_B \int_{B(Y,\rho)} \mathbb{E} \left[ g_\mathrm{X_Y,tx} \right] L(r_\mathrm{X_Y,tx}) d X_Y P_\mathrm{tx} \\
		& \overset{\mathrm{(b)}}{=} \frac{2 \pi}{\mu} \lambda_B \left( \int_{0}^{\rho} \frac{r}{r^\alpha + 1} dr \right) P_\mathrm{tx} = \gamma P_\mathrm{tx} 
	\end{split}		
	\end{equation} 
	where (a) is from the Campbell Theorem, (b) is by changing Cartesian coordinates to polar coordinates, and the last equality is achieved by defining $\gamma \triangleq \frac{2 \pi}{\mu} \lambda_B  \int_{0}^{\rho} \frac{r}{r^\alpha + 1} dr $.
\end{lemma}

\begin{lemma}
	The Laplace transform of the probability density function (PDF) of double fading random variable $g = g_1 g_2$, where $g_1, g_2 \sim \exp(\mu)$ is
	\begin{equation}
	\begin{split}
		\mathcal{L}_g(s) & = \mathbb{E}\left[ \exp(-sg) \right] =  \int_{0}^{\infty} e^{-s g} \int_{0}^{\infty} \frac{\mu^2}{t} e^{-\mu \left(t+\frac{g}{t}\right)} dt  dg  \\
		& = \int_{0}^{\infty} \frac{\mu^2 e^{-\mu t}}{s t + \mu} dt
	\end{split}
	\end{equation}
	where we use the fact that the PDF of $g$ is
	\begin{equation}
	f_{g}(g) = \int_{0}^{\infty} \frac{\mu^2}{t} e^{-\mu \left(t+\frac{g}{t}\right)} dt
	\end{equation}
	which can be derived according to the PDF of the product of two random variables.
	
\end{lemma}

\begin{lemma}
	Denoting $G = \omega_1 g_1 + \omega_2 g_2$ as the nonnegative weighted sum of two independent exponential random variables, $g_i \sim \exp(\mu_i),\ i = 1,2$, where $\omega_i \ge 0$, the cumulative distribution function (CDF) of $G$ is
	\begin{equation}
	\begin{split}
		F_G(g) = 1 - \frac{\tilde{\mu}_1}{\tilde{\mu}_1 - \tilde{\mu}_2} e^{-\tilde{\mu}_2 g} + \frac{\tilde{\mu}_2}{\tilde{\mu}_1 - \tilde{\mu}_2} e^{-\tilde{\mu}_1 g}, \ g \ge 0
	\end{split}
	\end{equation}
	where $\tilde{\mu}_i = \mu_i / \omega_i, i = 1,2$.
\end{lemma}

\begin{proof}	
	Please see Appendix B.
\end{proof}
Now, we are ready to derive the coverage probability. For the two cases according to whether $\Gamma(1-\beta) - \beta$ is negative or not, we have the following theorems.
\begin{theorem}
	When $0 \le \beta \le \frac{\Gamma}{\Gamma + 1}$, i.e., $\Gamma(1-\beta) - \beta \ge 0$, the SINR and SIR based coverage probabilities for the two scenarios are
	\begin{align}
		& \mathbb{P}(\mathrm{SINR} \ge \Gamma) \approx \xi_1 \xi_2 \zeta_1, \quad \mathbb{P}(\mathrm{SINR_u} \ge \Gamma) = \xi_1 {\xi}_\mathrm{2,u} \zeta_1 	\\
		& \mathbb{P}(\mathrm{SIR} \ge \Gamma) \approx \xi_1 \xi_2, \qquad  \mathbb{P}(\mathrm{SIR_u} \ge \Gamma) = \xi_1 {\xi}_\mathrm{2,u} 		
	\end{align}
	where 
	\begin{align*}
		\begin{split}
		& \xi_1 = \\
		& \exp \left\{ -\lambda_B \int_{X_{Y_0} \in B(Y_0, \rho)} \left( 1 - \int_{0}^{\infty} \frac{\mu^2 e^{-\mu t}}{a_\mathrm{X_{Y_0}} t + \mu} dt \right) d X_{Y_0} \right\}, \\
		\end{split}\\
		\begin{split}
		& \xi_2 = \\
		& \exp \left\{ -2 \pi \lambda_P \int_{r_0}^{R} \left( 1 - \frac{1}{1 + \Gamma \frac{r_0^\alpha + 1}{r^\alpha + 1}} \frac{1}{1 + \gamma \Gamma \frac{r_0^\alpha + 1}{r^\alpha + 1}} \right) r dr \right\},
		\end{split} \\
		\begin{split}
		& \xi_\mathrm{2,u} = \exp \left\{ -2 \pi \lambda_P \int_{r_0}^{R} \left( 1 - \frac{1}{1 + \Gamma \frac{r_0^\alpha + 1}{r^\alpha + 1}} \right) r dr \right\},
		\end{split} \\
		\begin{split}
		& \zeta_1 = \exp \left( -\frac{\mu \sigma^2 \Gamma }{L(r_0)P_\mathrm{tx}} \right),
		\end{split}
	\end{align*}
	$a_{X_{Y_0}} = \frac{\mu \eta \left[ \Gamma(1-\beta) - \beta \right]}{2 L(r_0)}  L(r_{X_{Y_0},tx}) L(r_{X_{Y_0},r}) $, and $\gamma$ is defined in Lemma 1.
\end{theorem}

\begin{proof}
	Please see Appendix C.
\end{proof}

Theorem 1 indicates that the coverage probabilities are the multiplications of two or three specific terms when $\beta$ is not greater than the threshold $\frac{\Gamma}{\Gamma + 1}$. In particular, $\xi_1$ corresponds to the effect of BTs around the typical PT, $\xi_2$ corresponds to the interference effect of atypical PTs and their surrounding BTs for scenario-1, $\xi_\mathrm{2,u}$ corresponds to the interference effect of atypical PTs for scenario-2, and $\zeta_1$ corresponds to the noise effect. Furthermore, the coverage probabilities in scenario-1 are upper bounded by the coverage probabilities in scenario-2 since $\xi_2<\xi_\mathrm{2,u}$.

\begin{theorem}
	When $\frac{\Gamma}{\Gamma + 1} < \beta \le 1$, i.e., $\Gamma(1-\beta) - \beta < 0$, the SINR and SIR based coverage probabilities of the two scenarios are
	\begin{align}
		& \mathbb{P}(\mathrm{SINR} \ge \Gamma) \approx  \frac{\tilde{\gamma}}{\tilde{\gamma}-1} \xi_3 \zeta_2 - 
		\frac{1}{\tilde{\gamma}-1} \xi_2 \zeta_1, \\
		& \mathbb{P}(\mathrm{SINR_u} \ge \Gamma) \approx  \frac{\tilde{\gamma}}{\tilde{\gamma}-1} \xi_\mathrm{3,u} \zeta_2 - \frac{1}{\tilde{\gamma}-1} \xi_\mathrm{2,u} \zeta_1, \\
		& \mathbb{P}(\mathrm{SIR} \ge \Gamma) \approx  \frac{\tilde{\gamma}}{\tilde{\gamma}-1} \xi_3 - 
		\frac{1}{\tilde{\gamma}-1} \xi_2, \\
		& \mathbb{P}(\mathrm{SIR_u} \ge \Gamma) \approx  \frac{\tilde{\gamma}}{\tilde{\gamma}-1} \xi_\mathrm{3,u} - \frac{1}{\tilde{\gamma}-1} \xi_\mathrm{2,u}    
	\end{align}
	where
	\begin{align*}
		\begin{split}
		&\xi_3 = \\
		& \exp \Bigg\{ -2 \pi \lambda_P \int_{r_0}^{R} \left( 1 - \frac{1}{1 + \frac{\Gamma (r_0^\alpha + 1)}{\tilde{\gamma}(r^\alpha + 1)}} \frac{1}{1 +  \frac{\gamma \Gamma(r_0^\alpha + 1)}{\tilde{\gamma}(r^\alpha + 1)}} \right) r dr \Bigg\}, 
		\end{split} \\
		\begin{split}
		&\xi_\mathrm{3,u} =  \exp \Bigg\{ -2 \pi \lambda_P \int_{r_0}^{R} \left( 1 - \frac{1}{1 + \frac{\Gamma (r_0^\alpha + 1)}{\tilde{\gamma}(r^\alpha + 1)}} \right) r dr \Bigg\}, 
		\end{split} \\
		\begin{split}
		& \zeta_2 = \exp \left( -\frac{\mu \sigma^2 \Gamma }{ \tilde{\gamma} L(r_0) P_\mathrm{tx} } \right),
		\end{split}
	\end{align*}
	$\tilde{\gamma} = -\left[ \Gamma(1-\beta) -\beta \right] \gamma > 0$, and $\gamma$ is defined in Lemma 1.
\end{theorem}

\begin{proof}
	Please see Appendix D.
\end{proof}

Different from the simple multiplication forms of Theorem~1 where each term in the multiplication corresponds to a specific effect, the coverage probabilities for $\beta$ greater than the threshold $\frac{\Gamma}{\Gamma + 1}$ are more complicated. Specifically, the coverage probabilities in Theorem 2 are expressed by weighted sums of the multiplications among $\xi_2$, $\xi_\mathrm{2,u}$, $\xi_3$, $\xi_\mathrm{3,u}$, $\zeta_1$, and $\zeta_2$ since we use a VT to approximate the BTs around the typical PT. In this case, the VT's effect is embodied by $\tilde{\gamma}$, $\xi_3$, $\xi_\mathrm{3,u}$, and $\zeta_2$. Furthermore, as ${P_\mathrm{tx}}/{\sigma^2}$ increases toward infinity, $\zeta_1$ and $\zeta_2$ approach 1, so that the SINR based coverage probabilities will converge to the SIR based coverage probabilities for an arbitrary $\beta$.

\section{Numerical Results} 

\begin{figure}[!tbp]
	\centering
	\includegraphics[width = 1\linewidth]{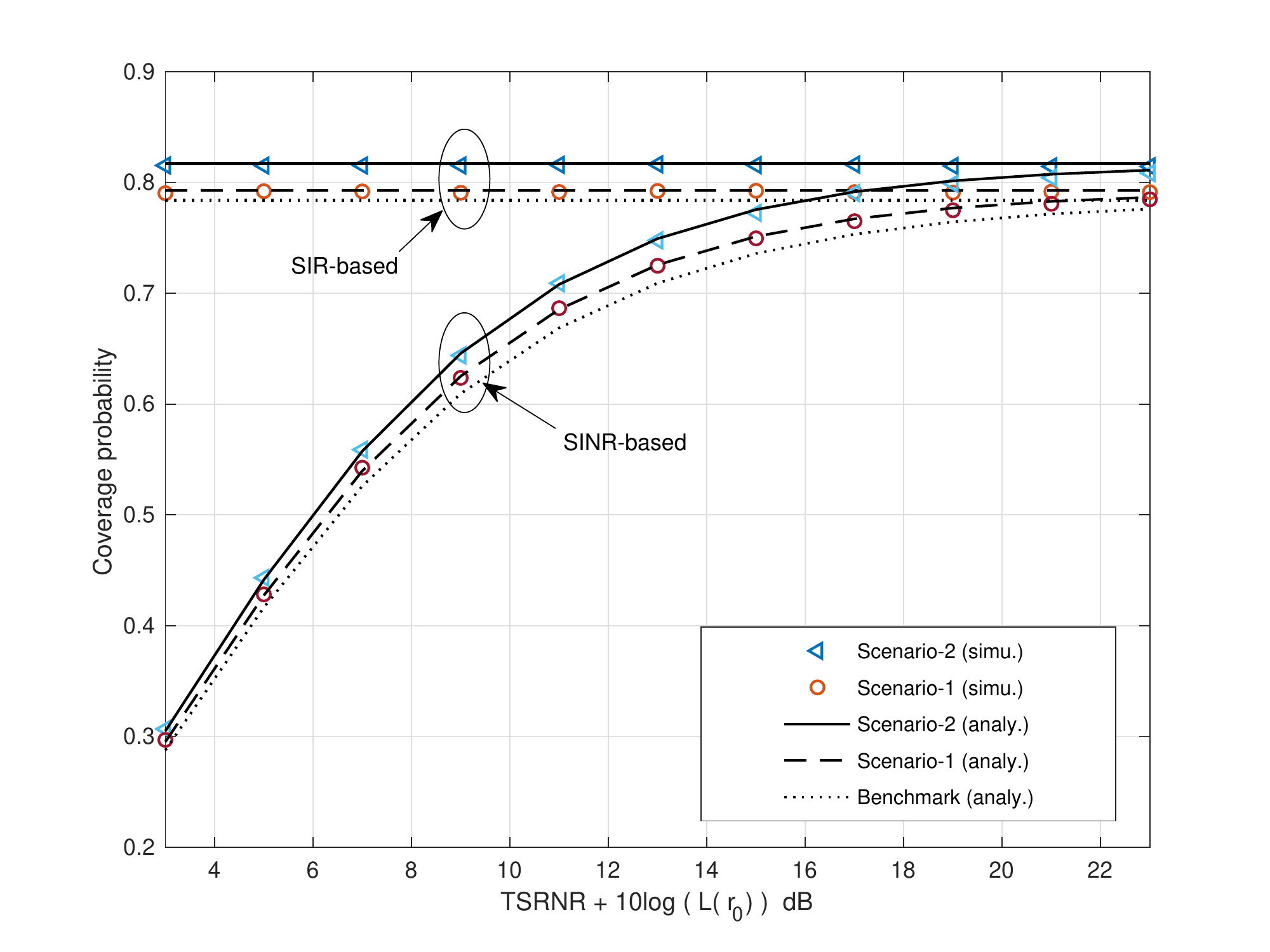}
	\caption{Coverage probability versus TSRNR deducting the absolute value of path loss. ($r_0 = 15, \Gamma = 3 \mathrm{dB}, \mu = 1, \beta = 0.8$)  }
	\label{fig_covprob_snr}
	\hfill
\end{figure}

\begin{figure}[!tbp]
	\centering
	\includegraphics[width = 1\linewidth]{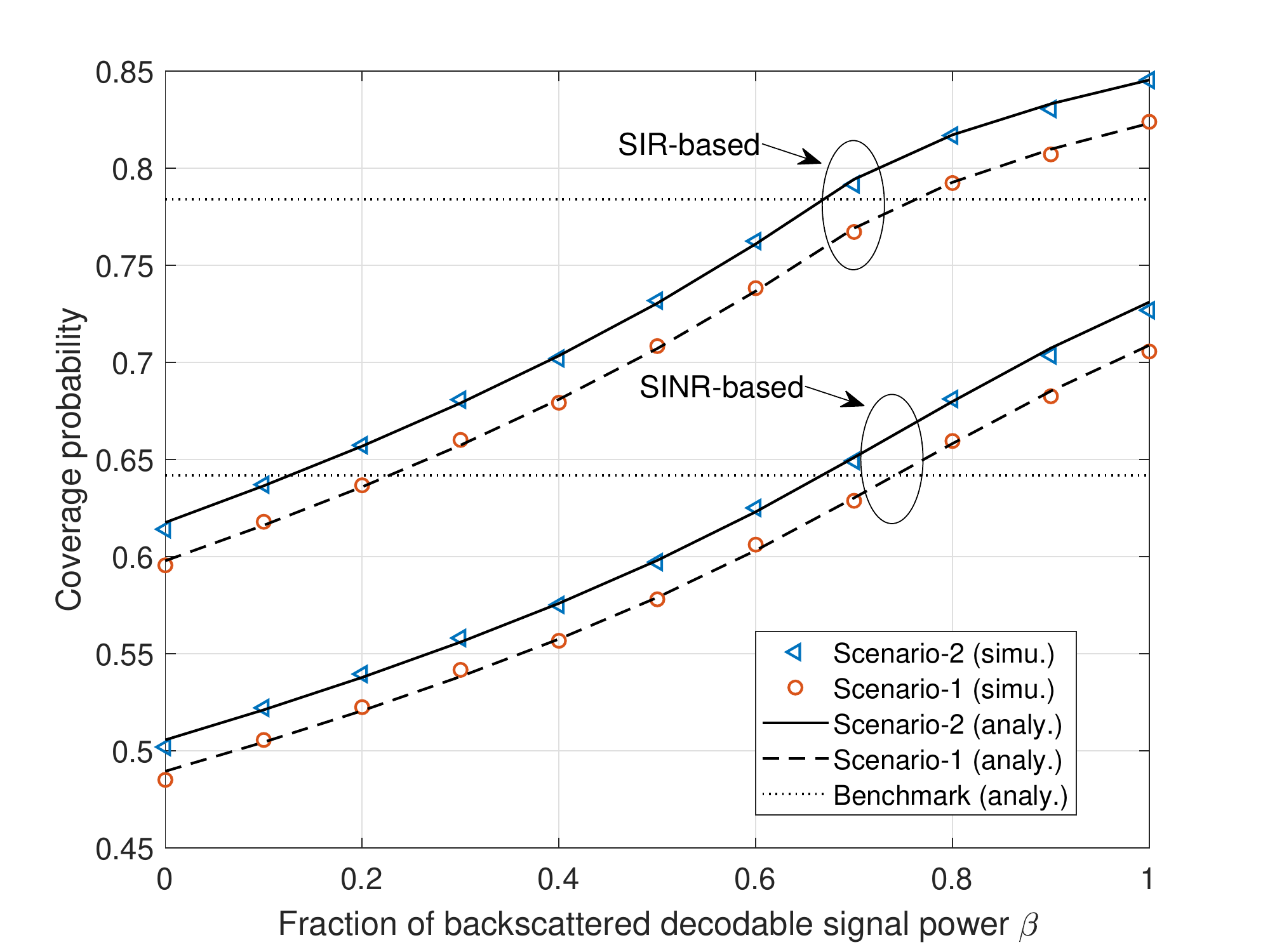}
	\caption{Coverage probability versus useful signal power ratio $\beta$. ($r_0 = 15, \Gamma = 3 \mathrm{dB}, \mu = 1, \mathrm{TSRNR} = 51 \mathrm{dB}$) }
	\label{fig_covprob_beta}
	\hfill
\end{figure}

The simulation results of both scenario-1 and scenario-2 are shown in Fig.~\ref{fig_covprob_beta}-\ref{fig_covprob_gamma} where we use Monte Carlo simulations with $50000$ independent system realizations to verify the analytical results. We set $\lambda_P = 2 \times 10^{-4}, \lambda_B = 0.1, \rho = 10, R = 100, \eta = 0.5, \alpha = 3.5$ by default. Other system settings are illustrated in the captions of the figures, where the distance metric unit is meter. Additionally, we also compare our results with the classic scenario (as a benchmark) where no BT exists (i.e., the interference received by the typical PR is only from atypical PTs). 	

The effect of TSRNR on the coverage probabilities is shown in Fig.~\ref{fig_covprob_snr}, where the horizontal axis represents the TSRNR after deducting the absolute value of path loss (which is $10\log_{10}(L(r_0)^{-1}) = 41$ dB). Clearly, the SIR based coverage probabilities are not affected by the TSRNR. However, the SINR based coverage probabilities gradually increase with the growth of TSRNR, finally converge to the SIR curves. With the listed system parameters, we observe that the TSRNR should be no less than 23 dB to make the SIR based coverage probabilities as accurate as the SINR based results. We set $r_0 = 15$ as a standard value for comparison and use TSRNR = 51 dB, i.e., the TSRNR after deducting the absolute value of path loss is 10 dB  (except for Fig.~\ref{fig_covprob_r0} where $r_0$ is a variable).

Fig. \ref{fig_covprob_beta} shows the coverage probabilities for different values of $\beta$. From 0 to 1, the value of $\beta$ indicates the fraction of backscattered signal power from the typical PT that can enhance the SINR and SIR at the typical PR. With the growth of $\beta$, the coverage probabilities of both scenario-1 and scenario-2 increase. In addition, when $\beta$ is larger than a certain value, the coverage probabilities of scenario-1 and scenario-2 exceeds the coverage probability of the benchmark scenario. These results correspond to the fact that the effect of backscattered signals has two sides: interference inducing and signal enhancing, i.e., the interference dominates when $\beta$ is less than the specific value, while the decodable signals dominate when $\beta$ is greater than that value.

\begin{figure}[!tbp]
	\centering
	\includegraphics[width = 1\linewidth]{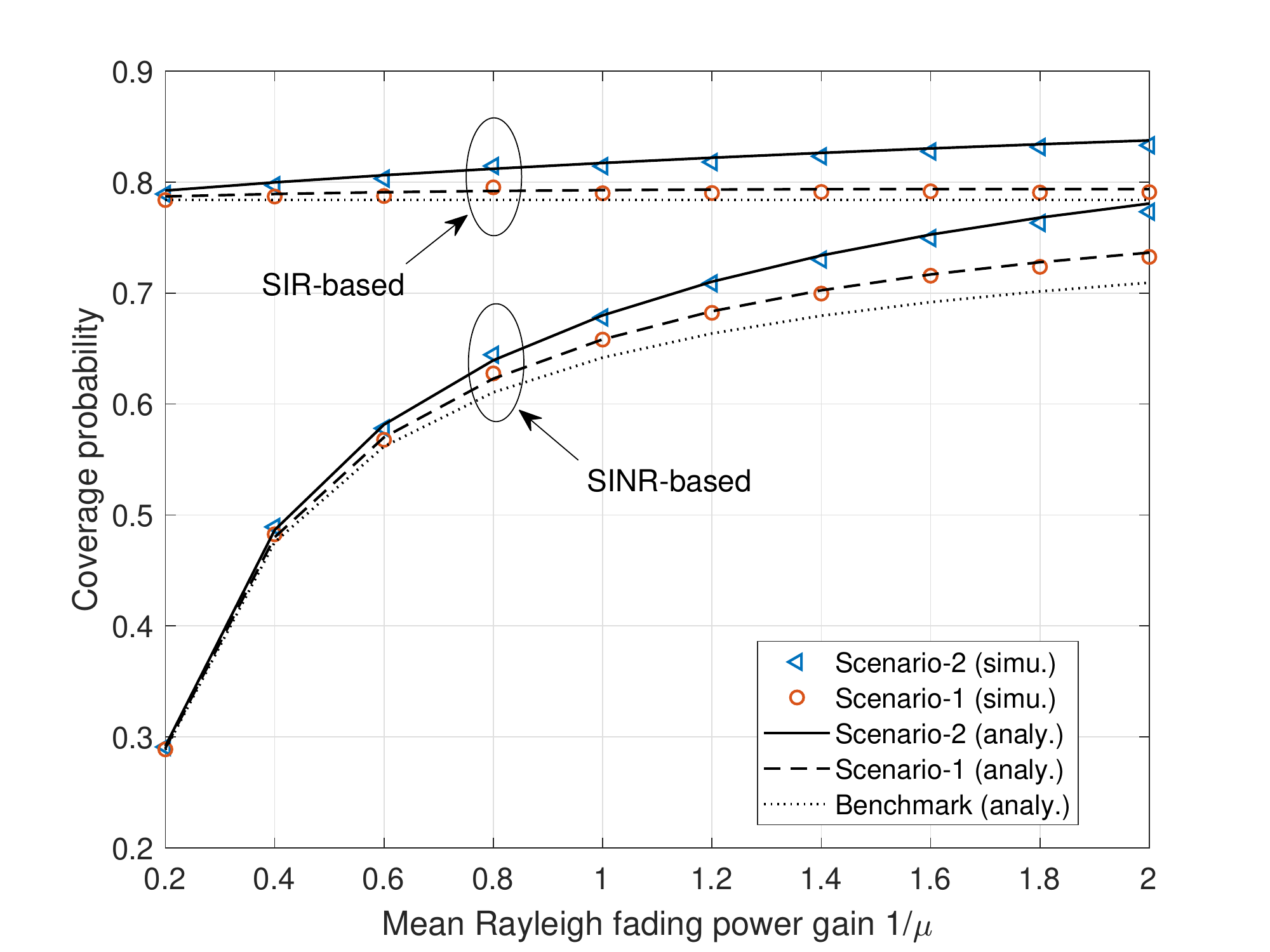}
	\caption{Coverage probability versus mean fading power gain. ($r_0 = 15, \Gamma = 3  \mathrm{dB}, \beta = 0.8, \mathrm{TSRNR} = 51 \mathrm{dB}$) }
	\label{fig_covprob_mean}
	\hfill
\end{figure}

\begin{figure}[!tbp]
	\centering
	\includegraphics[width = 1\linewidth]{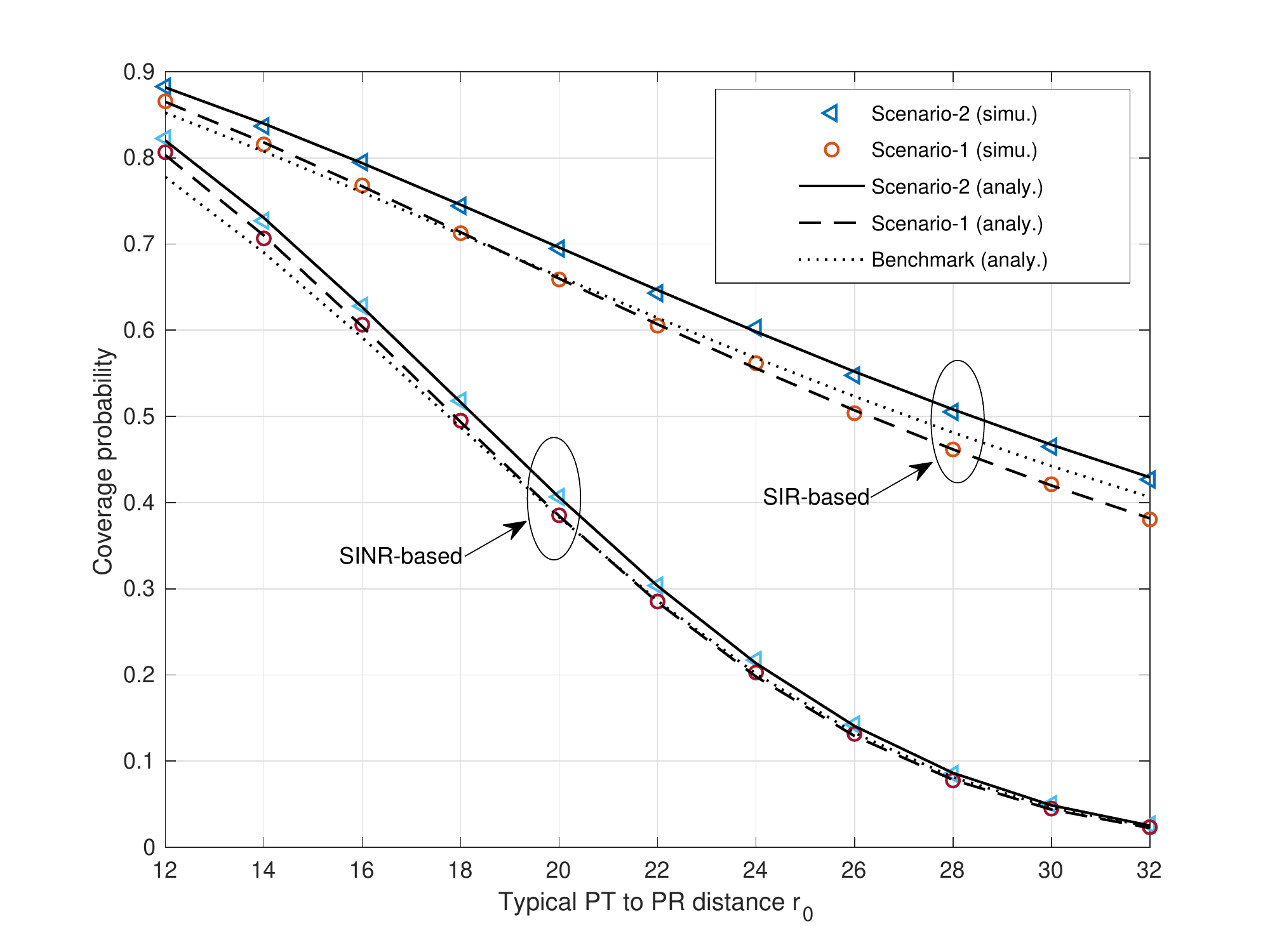}
	\caption{Coverage probability versus typical PT to PR distance $r_0$. ($\Gamma = 3 \mathrm{dB}, \beta = 0.8, \mu = 1, \mathrm{TSRNR} = 51 \mathrm{dB}$) }
	\label{fig_covprob_r0}
	\hfill
\end{figure}

\begin{figure}[!tbp]
	\centering
	\includegraphics[width = 1\linewidth]{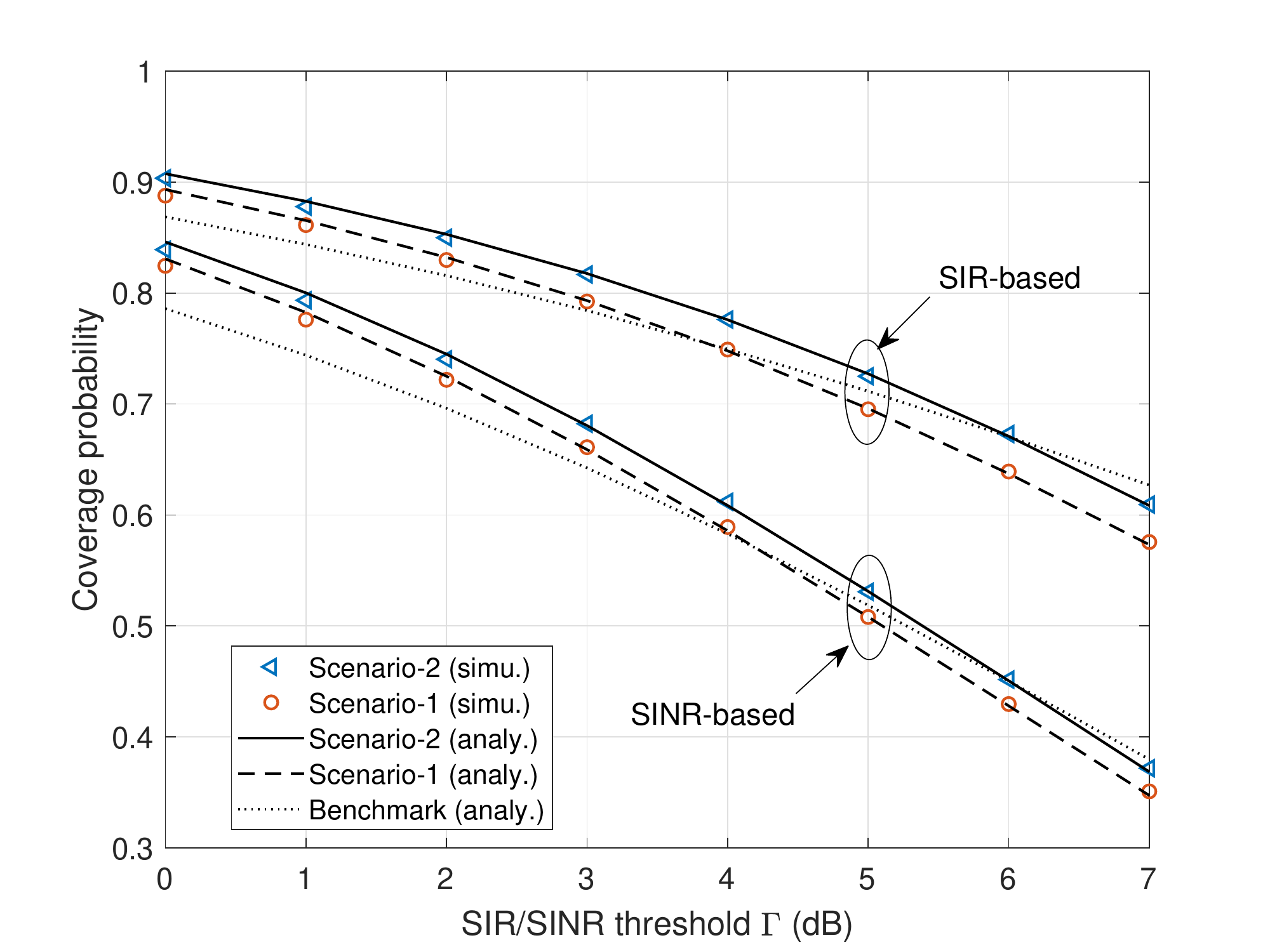}
	\caption{Coverage probability versus SIR threshold $\Gamma$. ($r_0 = 15, \beta = 0.8, \mu = 1, \mathrm{TSRNR} = 51  \mathrm{dB}$) }
	\label{fig_covprob_gamma}
	\hfill
\end{figure}

Typically, the mean power gain $1/\mu$ can be canceled in deriving the SIR based coverage probabilities for most of the network models (e.g., the benchmark scenario), if the channels are described as i.i.d. Rayleigh fading. However, in the considered two AmBC network scenarios, the channel power gains of the PT-BT path and the BT-PR path are multiplied due to the double fading effect, making the coverage probabilities more sensitive to the channel fading gain. As shown in Fig.~\ref{fig_covprob_mean}, the SIR based coverage probability of the benchmark scenario does not change with the fading power gain, but the SIR based probabilities for scenario-1 and scenario-2 increase with the growth of $1/\mu$. As the fading power gain grows, the SINR based coverage probabilities of all three scenarios increase, and tend to converge to the SIR based results since the noise becomes less significant. Furthermore, we observe that the increase of SINR and SIR based coverage probabilities are ranked as: scenario-2 $>$ scenario-1 $>$ benchmark scenario, as $1/\mu$ grows. This happens due to that 1) $\beta = 0.8$ is greater than the threshold $\frac{\Gamma}{\Gamma+1}$, thus the signal enhancing effect of the backscattered signals of a typical PT dominates (as Fig.~\ref{fig_covprob_beta} shows), and 2) there is less interference to a typical PR in scenario-2 than in scenario-1.

Fig.~\ref{fig_covprob_r0} shows that the coverage probabilities decrease with the increase of the typical PT to PR distance $r_0$. Moreover, as $r_0$ increases, BTs' signal enhancing effect becomes less significant than their interference effect, leading to the coverage probability of the benchmark scenario exceeds the coverage probability of scenario-1. Besides, the SINR based curves have steeper inclinations than the SIR based curves do for $r_0 < 26$, but gentler inclinations for $r_0 > 28$. This is because as the distance $r_0$ increases, the signal power from the typical PT and its surrounding BTs decreases exponentially, resulting in the SINR being dominated by noise.

It is observed in Fig. \ref{fig_covprob_gamma} that the coverage probabilities decrease as the SINR or SIR threshold $\Gamma$ increases. Moreover, with the growth of $\Gamma$, BTs' signal enhancing effect becomes less significant than their interference effect, leading to the coverage probability of the benchmark scenario exceeds the coverage probability of scenario-1. Besides, the SINR based curves always have greater inclinations than the SIR based curves do, due to the noise effect formulated by $\zeta_1$ and $\zeta_2$ in Theorem 2.


Furthermore, as shown in Fig.~\ref{fig_covprob_snr} - Fig.~\ref{fig_covprob_gamma}, the coverage probabilities of scenario-1 are always upper bounded by the results of scenario-2 as expected. We also notice that the simulation results almost perfectly match the analytical results, which indicates the VT approximation performs well.

\section{Conclusions}

We investigated new large-scale wireless network schemes, where AmBC nodes are involved in the conventional PPP-based communication model. In the new network scheme, the backscattered signals are regarded as either interference or decodable signals and we parameterized this effect with a fraction $\beta$. Depending on the value of $\beta$, SINR and SIR based coverage probabilities for two scenarios, where the BTs operate around either a single PT or all PTs, are derived in integral forms with the key system parameters. To make this work self-contained, we also provided an approach to estimate $\beta$. In summary, the coverage probability of our considered scheme lies in a wide range around the coverage probability of the conventional model, depending on the system settings, especially the value of $\beta$. Numerical results indicate that it is possible for the conventional communication system to keep satisfiable coverage probability while enabling backscatter communications for secondary systems if many AmBC nodes are involved in the large-scale wireless network.

\section*{Appendix A. \\ A Method to Estimate $\beta$ }

\begin{figure}[!tbp] 
	\centering
	\includegraphics[width = 0.9\linewidth]{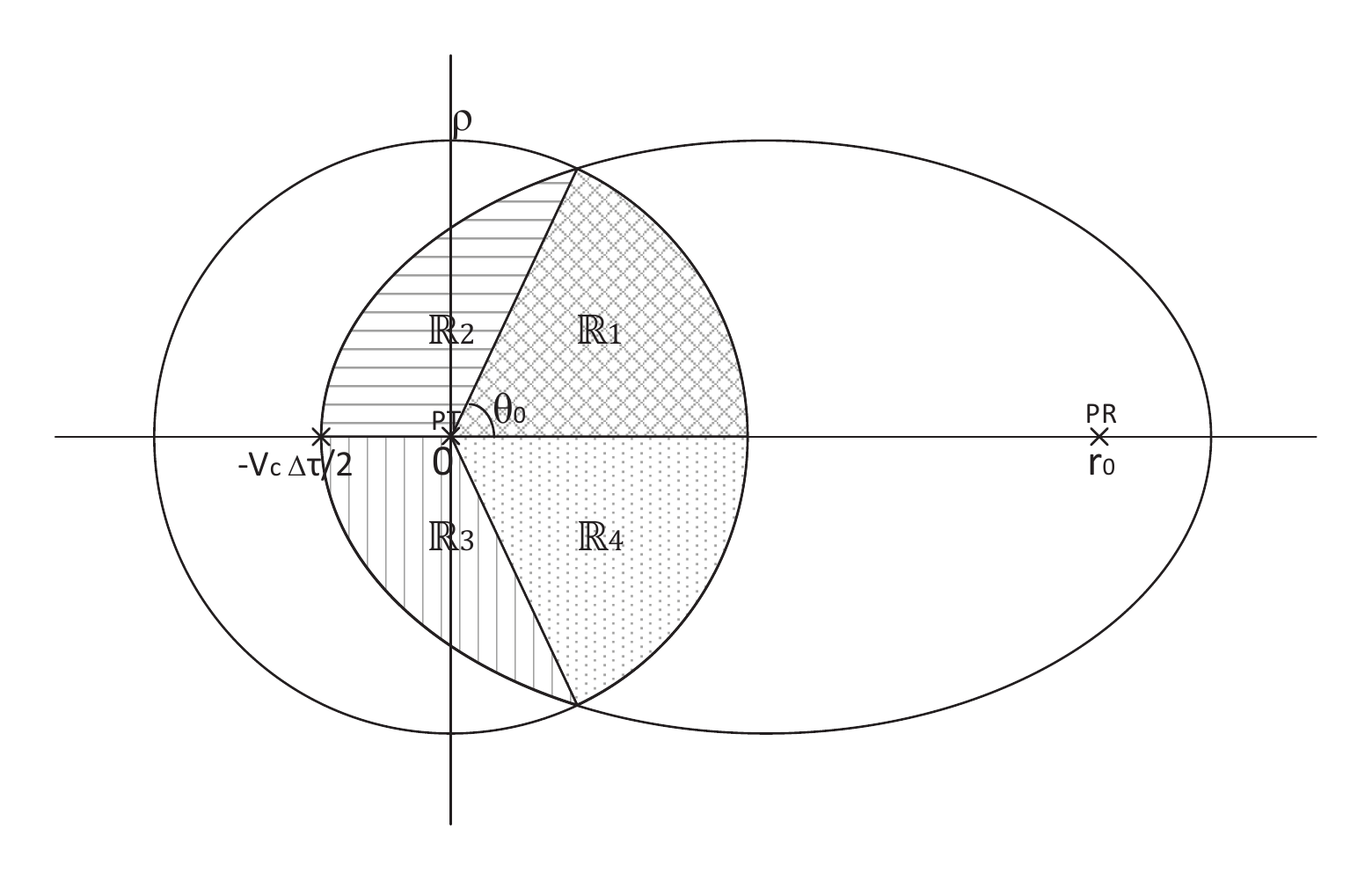} 
	\caption{Integration area in computing $\beta$} 
	\label{fig1}
	\hfill
\end{figure}

This part introduces an approach to compute $\beta$ with the available system parameters such as the distance between a typical PT and PR, the distribution range of the BTs, the symbol bandwidth of a PT, and the maximum tolerable delay for a typical PR.

First, denote the aggregated power of signals from the offsprings of a typical PT to be $\mathcal{S}_\mathrm{BT} = \mathcal{S}^{(1)}_\mathrm{BT} + \mathcal{S}^{(2)}_\mathrm{BT}$, where $\mathcal{S}_\mathrm{BT}^{(1)} = \beta \mathcal{S}_\mathrm{BT}$ represents the backscattered signal power that is not regarded as interference, while $\mathcal{S}_\mathrm{BT}^{(2)} = (1-\beta) \mathcal{S}_\mathrm{BT}$ is the counterpart that is regarded as interference. Since $\mathcal{S}^{(i)}_\mathrm{BT}, i = 1,2$ is the aggregated power of  signals from a relatively large  number of BTs, $\beta$ can be written approximately as the ratio between the two expected signal powers
\begin{equation} 
\beta = \frac{\mathcal{S}^{(1)}_\mathrm{BT}}{\mathcal{S}_\mathrm{BT}} \approx  \frac{ \mathbb{E}[\mathcal{S}_\mathrm{BT}^{(1)}]}{\mathbb{E}[\mathcal{S}_\mathrm{BT}]}.
\label{beta}
\end{equation}

Next, to compute $\mathbb{E} [ \mathcal{S}_\mathrm{BT}^{(1)} ]$, we assume a typical PT is located at the origin, the offspring BTs are located in the circle $B(\boldsymbol{0}, \rho)$, and the typical PR is located at $(0, r_0)$ {\color{black}as shown in Fig.~\ref{fig1}}. Then, we denote $\Delta \tau = k/\Delta B$ to be the maximum tolerable symbol delay at a typical PR, where $\Delta B$ represents the bandwidth of a symbol sent from the typical PT and $k$ is a scalar which indicates the ratio of the maximum tolerable delay to the symbol duration. Consequently, the BTs that can help to enhance the PR's receiving SINR and SIR are located in the intersection {\color{black}(shown as the shadow areas $\cup_{i=1}^4 \mathbb{R}_i$ in Fig.~\ref{fig1})} of a circular region $B(\boldsymbol{0}, \rho)$ and a ellipse region with the foci at $(0,0)$ and $(0,r_0)$, and a semi-major axes length of $l = \frac{1}{2} \left( {r_0} + {v_c \Delta \tau} \right)$, where $v_c = 3 \times 10^8 \mathrm{m/s}$ is the speed of radio waves. Then, expressions of the circle and the ellipse can be represented in polar coordinate\footnote{\color{black}We have a slight abuse of symbol $\theta$, which represents the channel phase in the system model. Here, $\theta$ represents the angel in a polar coordinate.} as
\begin{align}
r = \rho \label{circle},
\end{align}
and
\begin{align}
r = \frac{l(1-\epsilon^2)}{1-\epsilon \times \cos\theta} \label{ellipse}  
\end{align}
respectively, where $\epsilon = \frac{r_0}{2l}$ is the eccentricity of the ellipse. Equating (\ref{circle}) and (\ref{ellipse}), we obtain the intersection points $(\rho, \theta_0)$ and $(\rho, -\theta_0)$, where $\theta_0 = \cos^{-1} \left( \frac{1}{\epsilon} - \frac{l(1-\epsilon^2)}{\epsilon \rho} \right)$. Thus, we obtain (\ref{power1}) {\color{black}shown on top of next page},
\begin{figure*}
	\begin{equation} 
	\begin{split}
	&\mathbb{E}[\mathcal{S}_\mathrm{BT}^{(1)}] = \mathbb{E}\left[ \frac{\eta P_\mathrm{tx} }{2} \underset{X_{Y_0} \in \Phi_B(Y_0) \cap \mathbb{D} }{\sum} {g}_\mathrm{X_{Y_0},tx} {g}_\mathrm{X_{Y_0},rx}  { L(r_\mathrm{X_{Y_0},tx}) L(r_\mathrm{X_{Y_0},rx})   } \right]  \overset{\mathrm{(a)}}{=} \frac{\eta P_\mathrm{tx} }{2} \mathbb{E}[{g}_\mathrm{X_{Y_0},tx} {g}_\mathrm{X_{Y_0},rx}] \int_{B(\boldsymbol{0}, \rho) \cap \mathbb{D}} L(r_\mathrm{X_{Y_0},tx}) L(r_\mathrm{X_{Y_0},rx}) dX_{Y_0} \\
	& \overset{\mathrm{(b)}}{=}  \frac{\eta P_\mathrm{tx} }{ \mu^2} \int_{ \mathbb{R}_1\cup \mathbb{R}_2} L(r) L((r^2 + r_0^2 - 2 r r_0 \cos \theta)^{1/2}) r dr d\theta  \overset{\mathrm{(c)}}{=}  \frac{\eta P_\mathrm{tx} }{ \mu^2} \left( \int_{0}^{\theta_0} \int_{0}^{\rho} PL(r,\theta) r dr d\theta +  \int_{\theta_0}^{\pi} \int_{0}^{\frac{l(1-\epsilon^2)}{1-\epsilon  \cos\theta} } \mathrm{PL}(r,\theta) r dr d\theta  \right) \label{power1}
	\end{split}
	\end{equation}
\end{figure*}
where $\mathbb{D}$ denotes the ellipse region and $\mathrm{PL}(r,\theta) \triangleq L(r) L((r^2 + r_0^2 - 2 r r_0 \cos \theta)^{1/2})$ denotes the integrated path loss of the PT-BT-PR path for a BT at $(r,\theta)$. Specifically, (a) is from Campbell's Theorem, (b) is from changing the Cartesian coordinates to polar coordinates, and (c) is from separating the integration in region $\mathbb{R}_1$ and $\mathbb{R}_2$ shown in Fig.~\ref{fig1}. Similarly, we obtain
\begin{equation}
\mathbb{E}[\mathcal{S}_\mathrm{BT}]  =  \frac{\eta P_\mathrm{tx} }{ \mu^2} \int_{0}^{\pi} \int_{0}^{\rho} PL(r,\theta) r dr d\theta. \label{power2} 
\end{equation}
Therefore, (\ref{beta}), (\ref{power1}) and (\ref{power2}) yield
\begin{equation}
\beta \approx  \frac{\int_{0}^{\theta_0} \int_{0}^{\rho} PL(r,\theta) r dr d\theta +  \int_{\theta_0}^{\pi} \int_{0}^{\frac{l(1-\epsilon^2)}{1-\epsilon  \cos\theta} } \mathrm{PL}(r,\theta) r dr d\theta }{ \int_{0}^{\pi} \int_{0}^{\rho} PL(r,\theta) r dr d\theta}. 
\end{equation}

\begin{figure}[!tbp]
	\begin{subfigure}[b]{0.47\linewidth} 
		\centering
		\includegraphics[width = 1\linewidth]{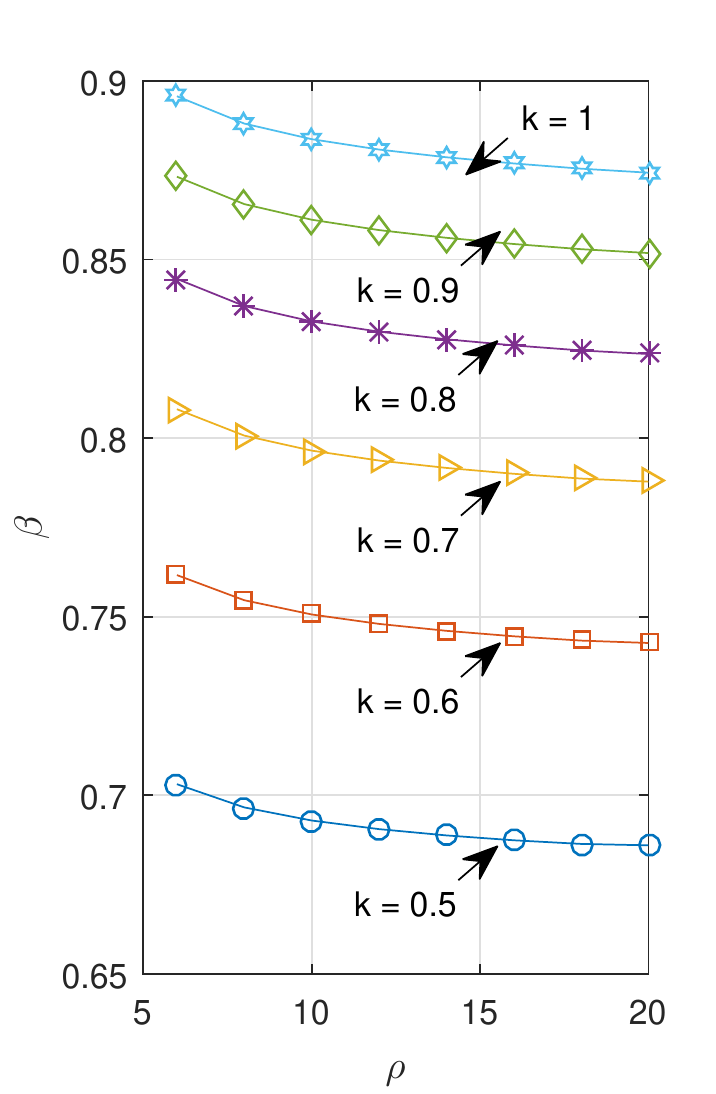} 
		\caption{$\beta$ versus $\rho$ ($r_0 = 25$)}
		\label{beta_rho}
		\hfill
	\end{subfigure}
	\begin{subfigure}[b]{0.47\linewidth} 
		\centering
		\includegraphics[width = 1\linewidth]{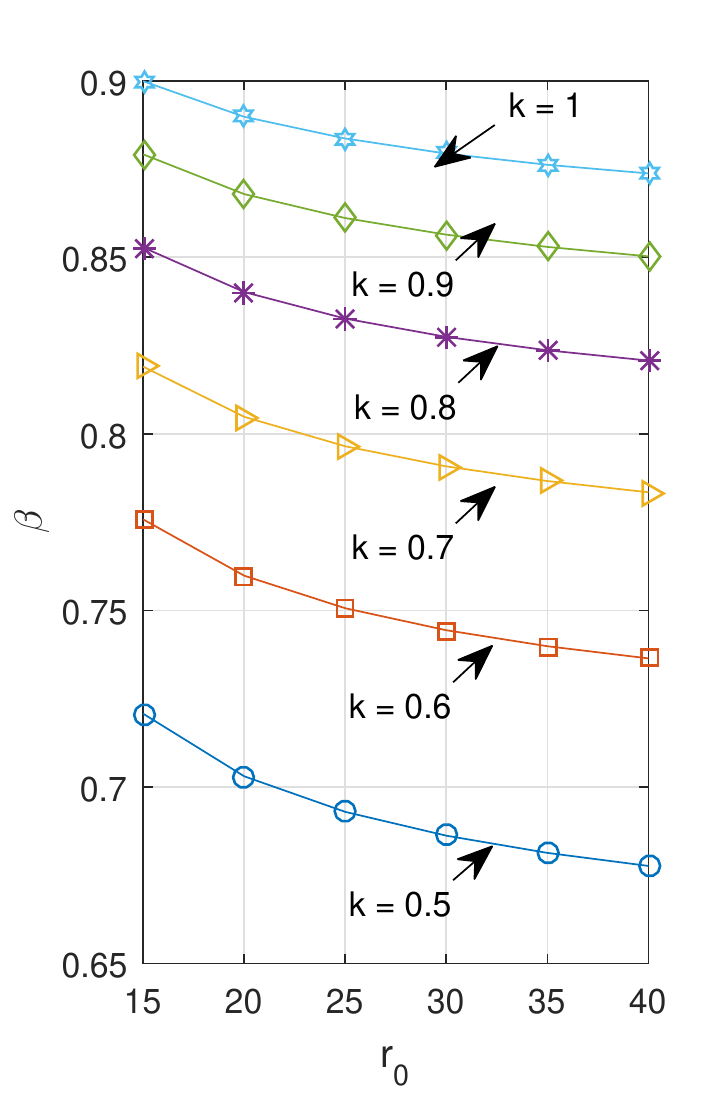} 
		\caption{$\beta$ versus $r_0$ ($\rho = 10$)}
		\label{beta_r0}
		\hfill
	\end{subfigure}	
	\vspace*{-0.3cm}	
	\caption{$\beta$ versus $\rho$ and $r_0$ ($\alpha = 3.5,\ \Delta B = 100$MHz )}	
	\label{fig_beta}
\end{figure}

Numerical results for $\beta$ versus $\rho$ and $r_0$ are shown in Fig.~\ref{beta_rho} and Fig.~\ref{beta_r0} respectively, where the settings are listed in the titles. These two figures show that fraction $\beta$ decreases with the growth of the radius $\rho$ and the typical PT-PR distance $r_0$. The former happens because increasing $\rho$ results in more BTs locating closer to the typical PR that can cause severer interference. When typical PT-PR distance $r_0$ increases, both $\mathcal{S}_\mathrm{BT}^{(1)}$ and $\mathcal{S}_\mathrm{BT}^{(2)}$ are reduced. However,  $\mathcal{S}_\mathrm{BT}^{(2)}$ decreases more than $\mathcal{S}_\mathrm{BT}^{(1)}$ since it is the aggregated power of signals from the offspring BTs that are farther to the typical PR (comparing with the other offspring BTs in the region $\mathbb{R}_i, i = 1,\ldots,4$, which contribute to $\mathcal{S}_\mathrm{BT}^{(1)}$, see Fig.~\ref{fig1}). Thus, a greater $r_0$ corresponds to a smaller $\beta$. Furthermore, Fig.~\ref{fig_beta} shows that $\beta$ increases with the growth of $k$ since a larger $k$ indicates a better delay tolerant capability at the typical PR such that only a smaller amount of BTs that are close enough to the typical PR can cause interference (i.e., the integration area $ \cup_{i=1}^{4} \mathbb{R}_i $ decreases).

\section*{Appendix B. Proof of Lemma 3}

Let $\tilde{g}_i = \omega_i g_i, i = 1,2$. Then, we recognize $\tilde{g}_i$ is exponentially distributed with mean $1/\tilde{\mu}_i = \omega_i / \mu_i$ and have
\begin{equation} 
\begin{split}
F_G(g) & = \mathbb{P}\left( \tilde{g}_1 \le g - \tilde{g}_2 \right) \\
& = \int_{0}^{g} \int_{0}^{g-t_2} \tilde{\mu}_1 e^{-\tilde{\mu}_1 t_1} \tilde{\mu}_2 	 e^{-\tilde{\mu}_2 t_2} d t_1 d t_2 \\
& = 1 - \frac{\tilde{\mu}_1}{\tilde{\mu}_1 - \tilde{\mu}_2} e^{-\tilde{\mu}_2 g} + \frac{\tilde{\mu}_2}{\tilde{\mu}_1 - \tilde{\mu}_2} e^{-\tilde{\mu}_1 g}
\end{split}
\end{equation}	
for $g \ge 0$.

\section*{Appendix C. Proof of Theorem 1}	
In the first scenario, substituting (\ref{I_BTPR})-(\ref{I_PTPR}) in (\ref{covprob_1_SINR}), we obtain equality (a) in (\ref{cov prob 2}) {\color{black}shown on the next page}. Then, (b) results from using Lemma 1 to replace the clustered BTs around atypical PTs with VTs, where $\tilde{g}_\mathrm{Y} \sim \exp(\mu)$ with $Y \in \Phi_P$ is the mean fading power gain of the channel between the VT at $Y$ (co-located with the PT) and the typical PR. (c) is from the complementary cumulative distribution function (CCDF) of exponential random variable $g_\mathrm{Y_0}$. Substituting $a_{X_{Y_0}} = \frac{\mu \eta \left[ \Gamma(1-\beta) - \beta \right]}{2 L(r_0)}  L(r_{X_{Y_0},tx}) L(r_{X_{Y_0},rx})$ to (c), we obtain (d). Next, (d) is written as the product of two expectations in (e) due to the independence between $X_{Y_0}$ and $Y$.
\begin{figure*}
	\begin{equation} 
	\begin{split}
	&\mathbb{P}(\mathrm{SINR} \ge \Gamma) \\
	&\overset{\mathrm{(a)}}{=} \mathbb{P} \Bigg\{  g_\mathrm{Y_0} \ge \frac{1}{L(r_0)} \Big[ \left( \Gamma(1-\beta)-\beta \right) \frac{\eta}{2} \sum_{X_{Y_0} \in \Phi_B(Y_0)} g_\mathrm{X_{Y_0},tx} g_\mathrm{X_{Y_0},rx} L(r_\mathrm{X_{Y_0},tx}) L(r_\mathrm{X_{Y_0},rx}) + \Gamma \sum_{Y \in \Phi_P} g_\mathrm{Y} L(r_\mathrm{Y}) \\ 
	& \quad + \Gamma \frac{\eta}{2} \sum_{Y \in \Phi_P} \sum_{X_Y \in \Phi_B(Y)} g_\mathrm{X_{Y},tx} g_\mathrm{X_{Y},rx} L(r_\mathrm{X_{Y},tx}) L(r_\mathrm{X_{Y},rx}) + \frac{\sigma^2 \Gamma}{P_\mathrm{tx}} \Big] \Bigg\} \\
	& \overset{\mathrm{(b)}}{\approx} \mathbb{P} \Bigg\{  g_\mathrm{Y_0} \ge \frac{1}{L(r_0)} \Big[ \left( \Gamma(1-\beta)-\beta \right) \frac{\eta}{2} \sum_{X_{Y_0} \in \Phi_B(Y_0)} g_\mathrm{X_{Y_0},tx} g_\mathrm{X_{Y_0},rx} L(r_\mathrm{X_{Y_0},tx}) L(r_\mathrm{X_{Y_0},rx}) + \Gamma \sum_{Y \in \Phi_P} ( g_\mathrm{Y} + \gamma \tilde{g}_\mathrm{Y}) L(r_\mathrm{Y}) + \frac{\sigma^2 \Gamma}{P_\mathrm{tx}} \Big] \Bigg\} \\
	& \overset{\mathrm{(c)}}{=} \mathbb{E} \Bigg\{  \exp \Bigg( \frac{-\mu}{L(r_0)} \Big[ \left( \Gamma(1-\beta)-\beta \right) \frac{\eta}{2} \sum_{X_{Y_0} \in \Phi_B(Y_0)} g_\mathrm{X_{Y_0},tx} g_\mathrm{X_{Y_0},rx} L(r_\mathrm{X_{Y_0},tx}) L(r_\mathrm{X_{Y_0},rx}) + \Gamma \sum_{Y \in \Phi_P} ( g_\mathrm{Y} + \gamma \tilde{g}_\mathrm{Y}) L(r_\mathrm{Y}) + \frac{\sigma^2 \Gamma}{P_\mathrm{tx}} \Big] \Bigg)  \Bigg\} \\
	& \overset{\mathrm{(d)}}{=} \mathbb{E} \Bigg\{  \exp \Bigg(  -\sum_{X_{Y_0} \in \Phi_B(Y_0)} a_\mathrm{X_{Y_0}} g_\mathrm{X_{Y_0},tx} g_\mathrm{X_{Y_0},rx}  - \frac{\mu \Gamma}{L(r_0)} \sum_{Y \in \Phi_P} ( g_\mathrm{Y} + \gamma \tilde{g}_\mathrm{Y}) L(r_\mathrm{Y})  \Bigg)  \Bigg\} \exp \left( \frac{-\mu \sigma^2 \Gamma}{L(r_0) P_\mathrm{tx}}\right) \\
	& \overset{\mathrm{(e)}}{=} \mathbb{E} \Bigg\{  \exp \Bigg(  -\sum_{X_{Y_0} \in \Phi_B(Y_0)} a_\mathrm{X_{Y_0}} g_\mathrm{X_{Y_0},tx} g_\mathrm{X_{Y_0},rx} \Bigg) \Bigg\}  \times \mathbb{E} \Bigg\{  \exp \Bigg(- \frac{\mu \Gamma}{L(r_0)} \sum_{Y \in \Phi_P} ( g_\mathrm{Y} + \gamma \tilde{g}_\mathrm{Y}) L(r_\mathrm{Y})  \Bigg)  \Bigg\} \zeta_1 \\
	\end{split}
	\label{cov prob 2}
	\end{equation} 	
\end{figure*}

Furthermore, we derive (\ref{cov prob 2_1}) and (\ref{cov prob 2_2}) {\color{black} on the next page,} where $g_\mathrm{X_{Y_0}} \triangleq g_\mathrm{X_{Y_0},tx} g_\mathrm{X_{Y_0},rx}$, and $\mathbb{D} = B(\boldsymbol{0}, r_0) - B(\boldsymbol{0}, R)$ is the distribution range of $Y$. (a) in (\ref{cov prob 2_1}) and (\ref{cov prob 2_2}) is from the probability generating functional (PGFL) of PPP \cite{haenggi2009interference}. Then, we obtain (b) in (\ref{cov prob 2_1}) by Lemma 2. By changing Cartesian coordinates to polar coordinates and calculating the expectation, (b) and (c) in (\ref{cov prob 2_2}) are obtained in sequence. Thus, we obtain $\mathbb{P}(\mathrm{SINR} \ge \Gamma) \approx \xi_1 \xi_2 \zeta_1$ by substituting (\ref{cov prob 2_1}) and (\ref{cov prob 2_2}) to (\ref{cov prob 2}). Applying similar derivations for the second scenario yields (\ref{cov prob 2s}) {\color{black}shown on the next page}. Setting $\frac{\sigma^2}{P_\mathrm{tx}} = 0$ (i.e., $\zeta_1 = 1$), we can obtain the SIR based coverage probabilities. 
\begin{figure*}
	\begin{equation}	
	\begin{split}
	&\mathbb{E} \Bigg\{  \exp \Bigg(  -\sum_{X_{Y_0} \in \Phi_B(Y_0)} a_\mathrm{X_{Y_0}} g_\mathrm{X_{Y_0},tx} g_\mathrm{X_{Y_0},rx} \Bigg) \Bigg\}  = \underset{ \Phi_B(Y_0)}{\mathbb{E}} \Bigg\{ \prod_{X_{Y_0} \in\Phi_B(Y_0)} \underset{ g_\mathrm{X_{Y_0}}}{\mathbb{E}} \left[ \exp \left( - a_\mathrm{X_{Y_0}} g_\mathrm{X_{Y_0}}  \right) \right] \Bigg\} \\
	&\overset{\mathrm{(a)}}{=} \exp \left\{ -\lambda_B \int_{B(Y_0,\rho)} \left( 1 - \underset{ g_\mathrm{X_{Y_0}} }{\mathbb{E}} \left[ e^{ - a_\mathrm{X_{Y_0}} g_\mathrm{X_{Y_0}} } \right] \right) d X_{Y_0} \right\} \overset{\mathrm{(b)}}{=} \exp \left\{ -\lambda_B \int_{ B(Y_0, \rho)} \left( 1 - \int_{0}^{\infty} \frac{\mu^2 e^{-\mu t}}{a_{X_{Y_0}} t + \mu} dt \right) d X_{Y_0} \right\} \\
	&= \xi_1
	\end{split}
	\label{cov prob 2_1}
	\end{equation}
\end{figure*}

\begin{figure*}
	\begin{equation}  
	\begin{split}
	& \mathbb{E} \Bigg\{  \exp \Bigg(- \frac{\mu \Gamma}{L(r_0)} \sum_{Y \in \Phi_P} ( g_\mathrm{Y} + \gamma \tilde{g}_\mathrm{Y}) L(r_\mathrm{Y})  \Bigg)  \Bigg\} = \mathbb{E} \Bigg\{  \prod_{Y \in \Phi_P} \exp \Bigg(- \frac{\mu \Gamma}{L(r_0)} ( g_\mathrm{Y} + \gamma \tilde{g}_\mathrm{Y}) L(r_\mathrm{Y})  \Bigg)  \Bigg\} \\
	& \overset{\mathrm{(a)}}{=} \exp \left\{ -\lambda_P \int_{\mathbb{D}} \Bigg( 1 - \underset{g_\mathrm{Y}, \tilde{g}_\mathrm{Y}}{\mathbb{E}} \Big[ e^{\frac{-\mu \Gamma}{L(r_0)} (g_\mathrm{Y} + \gamma \tilde{g}_\mathrm{Y}) L(r_\mathrm{Y})}   \Big]  \Bigg) d Y \right\} \overset{\mathrm{(b)}}{=} \exp \left\{ -\lambda_P \int_{0}^{2 \pi} \int_{r_0}^R \Bigg( 1 - \underset{}{\mathbb{E}} \Big[ e^{\frac{-\mu \Gamma L(r)}{L(r_0)} (g_\mathrm{Y} + \gamma \tilde{g}_\mathrm{Y}) }   \Big]  \Bigg) r dr d\theta \right\} \\
	& \overset{\mathrm{(c)}}{=} \exp \left\{ -2 \pi \lambda_P \int_{r_0}^{R} \left( 1 - \frac{1}{1 + \Gamma \frac{r_0^\alpha + 1}{r^\alpha + 1}} \frac{1}{1 + \gamma \Gamma \frac{r_0^\alpha + 1}{r^\alpha + 1}} \right) r dr \right\} = \xi_2
	\end{split}
	\label{cov prob 2_2}
	\end{equation}
\end{figure*}

\begin{figure*}
	\begin{equation} 
	\begin{split}
	&\mathbb{P}(\mathrm{SINR_u} \ge \Gamma) \\
	& = \mathbb{P} \Bigg\{  g_\mathrm{Y_0} \ge \frac{1}{L(r_0)} \Big[ \left( \Gamma(1-\beta)-\beta \right) \frac{\eta}{2} \sum_{X_{Y_0} \in \Phi_B(Y_0)} g_\mathrm{X_{Y_0},tx} g_\mathrm{X_{Y_0},rx} L(r_\mathrm{X_{Y_0},tx}) L(r_\mathrm{X_{Y_0},rx}) + \Gamma \sum_{Y \in \Phi_P} g_\mathrm{Y}  L(r_\mathrm{Y}) + \frac{\sigma^2 \Gamma}{P_\mathrm{tx}} \Big] \Bigg\}  \\
	&  = \mathbb{E} \Bigg\{  \exp \Bigg(  -\sum_{X_{Y_0} \in \Phi_B(Y_0)} a_\mathrm{X_{Y_0}} g_\mathrm{X_{Y_0},tx} g_\mathrm{X_{Y_0},rx} \Bigg) \Bigg\}  \times \mathbb{E} \Bigg\{  \exp \Bigg(- \frac{\mu \Gamma}{L(r_0)} \sum_{Y \in \Phi_P} g_\mathrm{Y}  L(r_\mathrm{Y})  \Bigg)  \Bigg\} \zeta_1 = \xi_1 \xi_\mathrm{2, u} \zeta_1
	\end{split}
	\label{cov prob 2s}
	\end{equation} 	
\end{figure*}

\section*{Appendix D. Proof of Theorem 2}
When $\Gamma(1-\beta)-\beta < 0$, the derivation of coverage probability for the first scenario can be started from the approximation (b) of (\ref{cov prob 2}). As we can notice, the equality (c) in (\ref{cov prob 2}) will not hold if $\Gamma(1-\beta)-\beta < 0$ since the exponential in (c) in (\ref{cov prob 2}) is not guaranteed to be non-positive. In this case, we further approximate the sum power from the clustered BTs around the typical PT as the power from a virtual transmitter located at $Y_0$. Then, from (b) in (\ref{cov prob 2}), the coverage probability is written as
\begin{equation}	
\begin{split}
& \mathbb{P}\left( \mathrm{SINR} \ge \Gamma \right) \\
& \overset{\mathrm{(a)}}{\approx} \mathbb{P} \Bigg\{  g_\mathrm{Y_0} + \tilde{\gamma} \tilde{g}_\mathrm{Y_0} \ge \frac{\Gamma}{L(r_0)} \sum_{Y \in \Phi_P} ( g_\mathrm{Y} + \gamma \tilde{g}_\mathrm{Y}) L(r_\mathrm{Y}) + \frac{\sigma^2 \Gamma}{P_\mathrm{tx}} \Bigg\} \\ 
& \overset{\mathrm{(b)}}{=} \frac{\mu}{\mu-{\mu}/{\tilde{\gamma}}} \mathbb{E} \Big[ \exp \Big( \frac{-\mu \Gamma}{\tilde{\gamma} L(r_\mathrm{Y_0})} \sum_{Y \in \Phi_P} (g_\mathrm{Y} + \gamma \tilde{g}_\mathrm{Y}) L(r_\mathrm{Y}) \Big) \Big] \zeta_2 \\
& \ \ \  - \frac{{\mu}/{\tilde{\gamma}}}{\mu-{\mu}/{\tilde{\gamma}}} \mathbb{E} \Big[ \exp \Big( \frac{-\mu \Gamma}{ L(r_\mathrm{Y_0})} \sum_{Y \in \Phi_P} (g_\mathrm{Y} + \gamma \tilde{g}_\mathrm{Y}) L(r_\mathrm{Y}) \Big) \Big] \zeta_1 \\
& \overset{\mathrm{(c)}}{=} \frac{\tilde{\gamma}}{\tilde{\gamma}-1} \xi_3 \zeta_2 - \frac{1}{\tilde{\gamma}-1} \xi_2 \zeta_1
\end{split}
\end{equation}  
where $\tilde{\gamma} = -\left[ \Gamma(1-\beta)-\beta \right] \gamma > 0$. (a) results from using Lemma 1 to replace the clustered BTs around the typical PT with a VT, where $\tilde{g}_\mathrm{Y_0} \sim \exp(\mu)$ is the mean power gain of the channel between the VT at $Y_0$ and the typical PR. Then, (b) is obtained by calculating the CCDF of $g_\mathrm{Y_0} + \tilde{\gamma} \tilde{g}_\mathrm{Y_0}$ with Lemma 3. Finally, we have (c) using PGFL of PPP, changing coordinates and calculating the expectation (with the same steps as in (\ref{cov prob 2_2})). 

For the second scenario, we derive the coverage probability in the same way and obtain
\begin{equation}	
\begin{split}
& \mathbb{P}\left( \mathrm{SINR_u} \ge \Gamma \right) \\
& {\approx} \mathbb{P} \Bigg\{  g_\mathrm{Y_0} + \tilde{\gamma} \tilde{g}_\mathrm{Y_0} \ge \frac{\Gamma}{L(r_0)} \sum_{Y \in \Phi_P} g_\mathrm{Y}  L(r_\mathrm{Y}) + \frac{\sigma^2 \Gamma}{P_\mathrm{tx}} \Bigg\} \\ 
& {=} \frac{\mu}{\mu-{\mu}/{\tilde{\gamma}}} \mathbb{E} \Big[ \exp \Big( \frac{-\mu \Gamma}{\tilde{\gamma} L(r_\mathrm{Y_0})} \sum_{Y \in \Phi_P} g_\mathrm{Y}  L(r_\mathrm{Y}) \Big) \Big] \zeta_2 \\
& \ \ \  - \frac{{\mu}/{\tilde{\gamma}}}{\mu-{\mu}/{\tilde{\gamma}}} \mathbb{E} \Big[ \exp \Big( \frac{-\mu \Gamma}{ L(r_\mathrm{Y_0})} \sum_{Y \in \Phi_P} g_\mathrm{Y}  L(r_\mathrm{Y}) \Big) \Big] \zeta_1 \\
& {=} \frac{\tilde{\gamma}}{\tilde{\gamma}-1} \xi_\mathrm{3,u} \zeta_2 - \frac{1}{\tilde{\gamma}-1} \xi_\mathrm{2,u} \zeta_1.
\end{split}
\end{equation} 
Setting $\frac{\sigma^2}{P_\mathrm{tx}} = 0$ (i.e., $\zeta_1 = \zeta_2 = 1$), we can obtain the SIR based coverage probabilities.

\bibliographystyle{IEEEtran}
\bibliography{reference_ambc}

\end{document}